\documentclass[sigconf, nonacm]{acmart}

\usepackage[noend]{algorithmic}
\usepackage{graphicx}
\usepackage{textcomp}
\usepackage{xcolor}
\usepackage{times}
\usepackage{epsfig,subfigure}
\usepackage[linesnumbered,ruled,noend]{algorithm2e}
\usepackage{booktabs}
\usepackage{multirow,multicol}
\usepackage{threeparttable}
\usepackage{comment}

\newcommand{\Paragraph}[1]{~\vspace*{-0.9\baselineskip}\\{\bf #1}}
\newcommand{\nop}[1]{}
\newtheorem{definition}{Definition}
\newtheorem{theorem}{Theorem}
\newtheorem{lemma}{Lemma}

\newcommand\vldbdoi{XX.XX/XXX.XX}
\newcommand\vldbpages{XXX-XXX}
\newcommand\vldbvolume{14}
\newcommand\vldbissue{1}
\newcommand\vldbyear{2022}
\newcommand\vldbauthors{\authors}
\newcommand\vldbtitle{\shorttitle} 
\newcommand\vldbavailabilityurl{URL_TO_YOUR_ARTIFACTS}
\newcommand\vldbpagestyle{plain} 

\newcommand{\picfolder}{./}

\begin{document}

% ****************** TITLE ****************************************

\title{WindGP: Efficient Graph Partitioning on Heterogenous Machines}

\author{ {Li Zeng},\ Haohan Huang,\ Binfan Zheng,\ Kang Yang,\ Shengcheng Shao,\and Jinhua Zhou, \ Jun Xie,\ Rongqian Zhao,\ Xin Chen}
\affiliation{%
\institution{Huawei Technologies Co., Ltd, China}
}
\email{{zengli43, huanghaohan, zhengbinfan1, yangkang18, shaoshengcheng}@huawei.com}
\email{{zhoujinhua1, xiejun1, zhaorongqian, chenxin}@huawei.com}

\begin{abstract}
Graph Partitioning is widely used in many real-world applications such as fraud detection and social network analysis, in order to enable the distributed graph computing on large graphs. 
However, existing works fail to balance the computation cost and communication cost on machines with different power (including computing capability, network bandwidth and memory size), as they only consider replication factor and neglect the difference of machines in realistic data centers. 
In this paper, we propose a general graph partitioning algorithm WindGP, which can support fast and high-quality edge partitioning on heterogeneous machines. 
WindGP designs novel preprocessing techniques to simplify the metric and balance the computation cost according to the characteristics of graphs and machines. 
Also, best-first search is proposed instead of BFS/DFS, in order to generate clusters with high cohesion. 
Furthermore, WindGP adaptively tunes the partition results by sophisticated local search methods.
Extensive experiments show that WindGP outperforms all state-of-the-art partition methods by 1.35$\times$$\sim$27$\times$ on both dense and sparse distributed graph algorithms, and has good scalability with graph size and machine number.
\end{abstract}

\maketitle

\pagestyle{\vldbpagestyle}
\begingroup\small\noindent\raggedright\textbf{PVLDB Reference Format:}\\
\vldbauthors. \vldbtitle. PVLDB, \vldbvolume(\vldbissue): \vldbpages, \vldbyear.\\
\href{https://doi.org/\vldbdoi}{doi:\vldbdoi}
\endgroup
\begingroup
\renewcommand\thefootnote{}\footnote{\noindent
	This work is licensed under the Creative Commons BY-NC-ND 4.0 International License. Visit \url{https://creativecommons.org/licenses/by-nc-nd/4.0/} to view a copy of this license. For any use beyond those covered by this license, obtain permission by emailing \href{mailto:info@vldb.org}{info@vldb.org}. Copyright is held by the owner/author(s). Publication rights licensed to the VLDB Endowment. \\
	\raggedright Proceedings of the VLDB Endowment, Vol. \vldbvolume, No. \vldbissue\ %
	ISSN 2150-8097. \\
	\href{https://doi.org/\vldbdoi}{doi:\vldbdoi} \\
}\addtocounter{footnote}{-1}\endgroup

\ifdefempty{\vldbavailabilityurl}{}{
	\vspace{.3cm}
	\begingroup\small\noindent\raggedright\textbf{PVLDB Artifact Availability:}\\
	The source code, data, and/or other artifacts have been made available at \url{\vldbavailabilityurl}.
	\endgroup
}

\section{Introduction}\label{sec:introduction}

Nowadays, many data can be modeled as graphs, such as financial transactions, social network, and road network.
With the strong representativity of graphs, researchers as well as industries can find out a lot valuable information by applying graph analysis algorithms (e.g., triangle counting \cite{HTC}, PageRank \cite{PageRank}, subgraph matching \cite{GSI,myFCS18,siep}, path finding \cite{KBQA} and Graph Neural Network \cite{GraphSage, KGCN, TransformerConv}) on these graph data.
In big-data era, distributed graph computing is widely used to process massive data in academic and industry (e.g., Huawei, Facebook, Google), supporting many important real-life applications such as financial fraud detection \cite{TigerGraph}, social network analysis \cite{DBLP:journals/itiis/MengCHCD18}, and online commodity recommendation \cite{GraphScope}.
Generally, the routine of distributed graph computing includes three parts: computation, communication, synchronization.
This is the BSP (Bulk Synchronous Parallel) paradigm \cite{BSP}, which is widely used in many frameworks such as Gemini \cite{Gemini} and PowerGraph \cite{PowerGraph}.
Figure \ref{fig:BSP} shows that BSP divides the whole process into several sequential supersteps, where a barrier exist between two supersteps.
In each superstep, on each machine, the communication occurs after the local computation.
In the end of each superstep, all machines are synchronized to ensure that they maintain the correct information before next superstep.
Thus, all machines need to wait for the slowest one, causing the long-tail effect \cite{DBLP:journals/pvldb/LuCYW14}.
As \cite{DBLP:journals/pvldb/LuCYW14} shows, with the growing data size and data complexity, the performance of distributed computing on billion-scale graphs is far from enough.

\vspace{-0.1in}
\begin{figure}[htbp]
	\centering
	\includegraphics[width=7cm]  {\picfolder 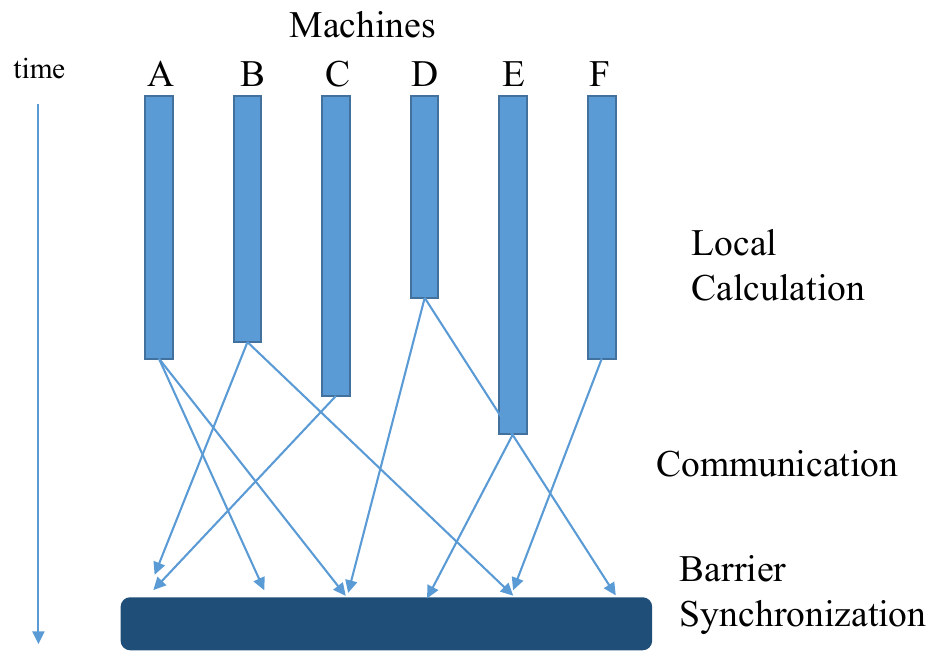}        	
    \vspace{-0.1in}     	
	\caption{The routine of BSP}      	
	\label{fig:BSP}  
\end{figure}
\vspace{-0.1in}

The quality of graph partition has great impact on the performance of distributed graph computing \cite{DBLP:journals/tpds/AkhremtsevSS20,NE}.
If the partition is highly skewed, the overall performance is dragged down by the long-tail effect.
Besides, if there are many connections between partitions, the communication cost will become the bottleneck.
Therefore, a partition can be called \emph{good} only when it is balanced and the number of cross-partition nodes/edges is small.

A running example of graph partition is given in Figure~\ref{fig:example}.
For vertex-centric partition, the vertices of the original graph $G$ is divided into three parts: $\{a,b,c\}$, $\{d,e\}$, $\{f\}$.
Similarly, the edges of $G$ are divided for edge-centric partition: $\{\overline{ab},\overline{bc}\}$, $\{\overline{de},\overline{ef}\}$, $\{\overline{cf}\}$.
The number of cross-partition edges in vertex-centric partition is two because $c$ and $e$ needs to communicate with $f$.
The number of cross-partition nodes in edge-centric partition is also two due to the synchronization of $c$ and $f$ between two partitions.

The problem of graph partition on homogeneous machines has been throughly studied, such as \cite{METIS, HDRF, NE, EBV}.
They mainly use two different metric (balance ratio and edge-cut/replication factor) and propose many optimization techniques in graph exploration and greedy selection.
All machines are viewed identically and partitions are randomly assigned to these machines.
%However, they do not distinguish different machines in realistic data centers, thus assign partitions to machines randomly.
%In consequence, graph partition on heterogeneous machines is rarely involved in the literature.

\vspace{-0.1in}
\begin{figure}[htbp]
	\centering
	\includegraphics[width=7cm]  {\picfolder 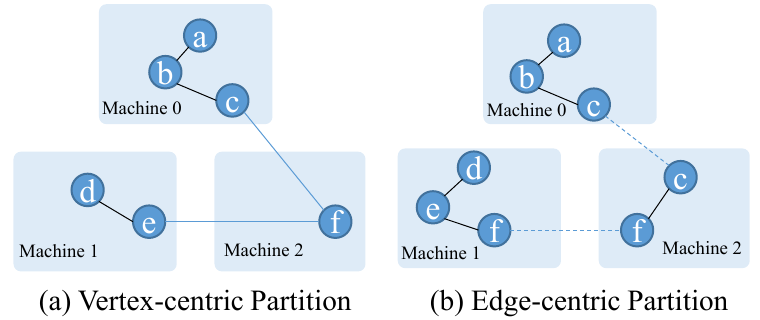}        	
	\vspace{-0.13in}     	
	\caption{An example of graph partition}      	
	\label{fig:example}  
\end{figure}
\vspace{-0.1in}

With the expiration of Moore Law \cite{DBLP:journals/cacm/Edwards21a}, heterogeneous computing is favored by more and more organizations because various heterogeneous hardwares have mushroomed all over the area.
On the one hand, machines with different configurations (e.g., CPU cores/frequency, memory capacity, network bandwidth) exist due to the evolution of hardwares \cite{GrapH,HeterCompPart}.
In all data centers, machines are purchased in different decades with different configurations.
On the other hand, heterogeneous computing is the trend in super computing area and real-life applications due to better performance and cheaper price.
In Telecom field, the machine resource in a region of carrier companies is limited and heterogeneous, where graph algorithms on billions of edges need to be performed locally to search reachable paths or find the cause of network faults.
Note that telecom data can not be transfered to cloud computing due to strict data privacy, thus they can only be processed by several low-memory edge servers or even personal computers.
As a result, distributed graph computing on heterogeneous machines \cite{DBLP:journals/grid/SulaimanHLWT21,DBLP:conf/hpcc/MasoodMRK15,DBLP:conf/bigdataconf/XueYHD15} is becoming more and more important nowadays.
This implies higher requirement and new challenges on graph partition methods.

Unfortunately, existing solutions of graph partition do not support heterogeneous machines.
Though they can be simply modified to adapt to heterogeneous situations by adding constraints of memory capacity, the results are terrible when measuring quality or running distributed graph algorithms \cite{DBLP:conf/bigdataconf/XueYHD15}.
Previous algorithms (except for \cite{GrapH}) can not generate high-quality graph partition on heterogeneous machines because they do not utilize the characteristics of different machines.
Besides, they can not achieve a good balance between computing cost and communication cost.
Our techniques address these problems and boost the performance of distributed graph computing on various machines.

Note that the optimization of graph partition is orthogonal to the accelerative techniques of distributed graph algorithms \cite{PageRank,HTC}. 
Though a graph algorithm has different implementations, the performance of each implementation can be further boosted by the improvement of the partition quality.
Specifically, the best partition strategies of different graph algorithms may vary \cite{PowerLyra,Gemini}, but a common base is essential for fast application and extension in realistic systems \cite{PowerGraph,Gemini}.
Algorithm-specific optimization can be explored base on our generic strategy, which is not in the scope of this paper.
Besides, though the computation cost of various tasks are different, it is proportional to the number of nodes or edges \cite{DBLP:journals/pvldb/VermaLSG17}.
Thus, our solution can be applied to different tasks when given the corresponding amount of work per node/edge.

Our contributions can be concluded below:
\begin{itemize}
\item We design a scalable framework of graph partitioning, which can support fast and high-quality edge partitioning on heterogeneous machines. 
\item Graph-oriented preprocessing techniques are utilized to simplify the metric and balance the computation cost according to the characteristics of graphs and machines. 
\item Best-first search scheme is proposed instead of BFS/DFS, in order to generate partitions with high cohesion. 
\item Sophisticated subgraph-local search methods are adopted to tune the partition results adaptively.
\item Extensive experiments on both dense and sparse distributed graph algorithms show that WindGP outperforms all state-of-the-art partition methods by 1.35$\times$$\sim$27$\times$.
\end{itemize}

The rest of this paper is organized as follows.
Section \ref{sec:background} gives the formal problem definition and reviews state-of-the-art implementations including both vertex-centric and edge-centric methods.
Section \ref{sec:algorithm} presents our framework and optimization techniques.
The extensions of WindGP is discussed in Section \ref{sec:extension}.
The evaluation of our solution is in Section \ref{sec:experiment}.
Finally, Section \ref{sec:conclusion} concludes the paper.

\section{Background}\label{sec:background}

In this section, we first present the formal definition of our problem, then list the related work.
The common notations within the paper are summarized in Table \ref{tab:notations}.

\subsection{Problem Definition} \label{sec:problem}

\begin{definition}\label{def:graph} \textbf{(Graph)}
	A graph is denoted as $G=\{V,E\}$, where $V$ is the set of vertices; $E \subseteq V \times V$ is the set of undirected edges.
	$V(G)$ and $E(G)$ are used to denote vertices and edges of graph $G$, respectively.
    Note that $\overline{uv}$ is equivalent to $\overline{vu}$ in undirected graphs.
\end{definition}

\begin{definition}\label{def:subgraph} \textbf{(Subgraph)}
	Given a graph $G=\{V,E\}$, a subgraph of $G$ is denoted as $G^{\prime}=\{V^{\prime},E^{\prime}\}$, where vertex sets $V^{\prime}$ and edge sets $E^{\prime}$ in $G^{\prime}$ are subsets of $V$ and $E$, respectively, denoted as $V^{\prime} \subseteq V$ and $E^{\prime} \subseteq E$.
\end{definition}

\begin{definition}\label{def:epart}\textbf{($p$-edge Partition)}
    Given a graph $G=\{V,E\}$, $p$ is the number of partitions, then the $p$-edge partition of $G$ is denoted as a sequence of subgraphs $EP(G)=\{G_{i}, \forall i\in [1,p]\}$, such that: \\
    (1) $G_{i}$ is a subgraph of $G$ and $\forall u\in V(G_{i}), \exists \overline{uv}\in E(G_{i})$; \\
    (2) $\bigcup_{i} E(G_{i})=E(G)$ and $E(G_{i}) \bigcap E(G_{j})=\emptyset, \forall i\neq j$. \\
    For the $i$-th partition, $V(G_{i})$ and $E(G_{i})$ can also be simplified as $V_{i}$ and $E_{i}$.
\end{definition}

% Assume that the $i$-th partition is assigned to the $i$-th machine 
\begin{definition}\label{def:problem}\textbf{(Problem Statement)}
    Given a graph $G=\{V,E\}$ and $p$ machines $Machine_{i}=\{M_{i}, C_{i}^{node}, C_{i}^{edge}, C_{i}^{com}\}, \forall i\in[1,p]$, the heterogeneous-machine graph partition problem is to find the best edge partition $EP(G)$ that minimizes the total cost $TC$, where \\
    (1) $TC=max_{i}\{T_i\}, where T_i=T_{i}^{cal}+T_{i}^{com}$; \\
    (2) $M_{i}\geq M^{node}\times |V_{i}|+M^{edge}\times |E_{i}|$, $M_{i}$ is the memory size of the $i$-th machine while $M^{node}$ and $M^{edge}$ are the memory occupation of a node and an edge, respectively; \\
    (3) $T_{i}^{cal}=C_{i}^{node}\times |V_{i}| + C_{i}^{edge}\times |E_{i}|$, where $C_{i}^{node}$ and $C_{i}^{edge}$ are the computing cost of a node and an edge, respectively; \\
    (4) $T_{i}^{com}=\sum_{v\in V_{i}} \sum_{j\neq i}^{v\in V_{j}} (C_{i}^{com}+C_{j}^{com})$, where $C_{i}^{com}$ is the communication cost of the $i$-th machine; \\
\end{definition}

This paper aims to provide fast and high-quality solutions for the problem of heterogeneous-machine graph partition (Definition~\ref{def:problem}).
Obviously, this problem is edge-centric, i.e., any edge of $G$ can only exist in a single partition, but this does not hold for vertex.
%The details of machine parameter setting ($M_i$, $C_i^{node}$, $C_i^{edge}$, $C_i^{com}$) are in Section \ref{sec:experiment}.
%In Pregel-like systems (e.g., PowerGraph and Gemini), computations  also included node computing and edge accumulation
%For example, in vertex-centric and edge-centric computing systems, $C_{i}^{edge}$ and $C_{i}^{node}$ can be set to 0, respectively.
As heterogeneous machines are rather different from homogeneous configurations, we propose a new metric $TC$ to measure the total time cost of each kind of edge partition.
$TC$ considers both computing cost and communication cost, while the traditional replication factor \cite{DBH,NE,EBV} ($RF=\frac{\sum_{u\in G}{|S(u)|}}{|V(G)|}$) only considers the communication cost ($S(u)$ represents the set of partitions that $u$ exists).

\Paragraph{Equivalence of metrics}.
The new metric $TC$ theoretically corresponds to the load balance and the $RF$ metric.
%In fact, $TC$ is equivalent to $max_{i}(T_{i}^{cal})+max_{i}(T_{i}^{com})$, which implies that during each superstep, the local computation is all finished within each machine before the cross-machine communication.
In homogeneous cases, $T_{i}^{cal}$ is proportional to $|V_{i}|$ or $|E_{i}|$ according to the design of distributed framework.
Besides, assuming $C_i^{com}=1\ \forall i$ and $|V(G)|$ is fixed, we can deduce  $\sum_{i}T_{i}^{com}=\sum_{u\in G}{|S(u)|\times(|S(u)|-1)}=\Theta(RF^2)$, which is consistent with $RF$ and both of them can be used to measure the communication cost.
%RF*|V(G)|>1 in all cases
Comparative experiments on both dense algorithms (PageRank \cite{PageRank}, Triangle \cite{HTC}) and sparse algorithms (BFS \cite{BFS}, SSSP \cite{SSSP}) shows that $TC$ is proportional to the distributed running time with <10\% error (see Table \ref{tab:verify-tc}).
Furthermore, given a specific graph algorithm, the computation cost of a node or an edge is the same for different nodes/edges.
Therefore, $TC$ is used as the metric of partition quality in heterogeneous environments.

\Paragraph{NP-Hardness}.
The problem of heterogeneous-machine graph partition can be proved to be NP-hard for any $p\geq 2$.
First, let $\alpha^{\prime}$ be the balance ratio in \cite{NE}, the problem of $p$-edge partition with minimal $RF$ (i.e., $MIN$-$RF(p,\alpha^{\prime})$) is NP-hard \cite{DBLP:conf/kdd/BourseLV14, NE}.
Second, the problem of $MIN$-$RF(p,\alpha^{\prime})$ can be reduced to our problem because it is a special case that all machines are homogeneous and the memory size is set to $\frac{\alpha^{\prime} |E|}{p}$ (let $M^{node}=0$ and $M^{edge}=1$).

\Paragraph{Quantification of Machine Resource}.
Similar to \cite{HaSGP,DBLP:conf/ic2e/GarraghanTX13,DBLP:journals/tompecs/HerbstBKOEKEKBA18,DBLP:journals/simpra/ZakaryaG19}, the resources of machines can be quantified by relative rates:
\begin{itemize}
\item memory capacity:  let $Mem_i$ GB be the memory size of each machine, $M_{i}$ is calculated by $\frac{10^9\times Mem_i}{4\times gcd(\{Mem_i\})}$. 
\item compute ability: each machine multiplies a float-point with an integer, repeat this process many times and yield the averaged $FPTime_{i}$, then $C_{i}^{node}$ is $\frac{FPTime_i}{gcd(\{FPTime_i\})}$; $C_{i}^{edge}$ is computed by two operations (sum and multiplication) and set to $\frac{FPTime^{\prime}_i}{gcd(\{FPTime_i\})}$.
\item network bandwidth: each machine sends/receives 4KB data many times, the averaged time cost is $COTime_{i}$, thus $C_{i}^{com}=\frac{COTime_{i}}{1024\times gcd(\{FPTime_i\})}$. 
\end{itemize}
In this paper, we assume each node occupies 32 bits and there is no attribute for node/edge computing. 
Thus, $M^{node}$ is set to $\frac{1}{gcd(\{Mem_i\})}$ and $M^{edge}=2\times M^{node}$.
If attributes exist in computing, $FPTime$ and $M^{node}$ should be multiplied with an appropriate number.
Though different graph algorithms have various computing cost for a node or an edge, the total cost is proportional to the number of nodes or edges, which implies that the quantification of machine resource can be simply combined with specific graph algorithms.
% cal cost and comm cost need to sum, thus they must represent the real scale of time 

Assume there are three heterogeneous machines and the configuration of machines is $Machine_0=\{7,0,1,1\}$, $Machine_1=\{7,0,2,2\}$ and $Machine_2=\{5,0,1,1\}$ ($C_{i}^{node}$ is set to 0 for simplicity).
Let $M^{node}$ be 1 and $M^{edge}$ be 2.
In Figure~\ref{fig:example}(b), according to Definition~\ref{def:problem}, a valid edge partition is $\{\overline{ab},\overline{bc}\}$ on $Machine_0$, $\{\overline{de},\overline{ef}\}$ on $Machine_1$, and $\{\overline{cf}\}$ on $Machine_2$.
The computing cost and communication cost of these machines are $\{2,2\}$, $\{4,3\}$ and $\{1,5\}$, respectively.
The corresponding $TC$ is $7$, while the $RF$ value is $1.33$.
In contrast, another solution is $\{\overline{ab}\}$ on $Machine_0$, $\{\overline{bc},\overline{cf}\}$ on $Machine_1$, and $\{\overline{de},\overline{ef}\}$ on $Machine_2$, and the corresponding $TC$ and $RF$ are $10$ and $1.33$.
Obviously, the $RF$ remains unchanged, while the $TC$ becomes larger when adjusting the assignment of partitions to machines.
Therefore, the existing partition methods that work well on homogeneous machines can generate terrible results on heterogeneous machines.

In this paper, we focus on the partition algorithm on a single machine, i.e., the entire graph $G$ is partitioned by a single machine and the partition results are moved to heterogeneous machines for distributed running.
Without loss of generality, we assume that at least one edge partition is feasible for given graph $G$ and machines, and partition $G_{i}$ is assigned to machine $Machine_{i}$.
Though our solution can be easily extended to process directed graphs, vertex/edge labels or label sets, that is not our focus.
Complex optimizations of distributed system (e.g., aggregated communication and overlapping of computation phases) are algorithm-specific, which will be explored in the future. 
Unless otherwise specified, we use $u$, $N(u)$, $deg(u)$, $num(L)$, and $|A|$  to denote a vertex, the neighbor set of $u$, degree of $u$, the number of currently valid elements in set $L$, and the size of set $A$, respectively.

\vspace{-0.1in}
\begin{table}[htbp]
	\small
	\caption{The relationship between $TC$ and distributed running time (unit:s)}
	\label{tab:verify-tc}
	\vspace{-0.1in}
	\begin{threeparttable}
		%\small
		\centering
		\setlength{\tabcolsep}{2.5mm}{
			\begin{tabular}{crrrrr}
				% this is 9-machine cluster, HDRF is better than NE on TW
				% but on 30-machine cluster, NE is better than HDRF
				\toprule
				Sol. & $TC$ & PageRank & Triangle & SSSP & BFS \\
				\midrule
                HDRF \cite{HDRF} & 2.7G & 2.3K & 0.7K & 1K &  0.11K    \\
                NE \cite{NE} & 5.6G & 4.4K & 1.5K & 1.8K & 0.2K   \\
				\bottomrule
			\end{tabular}
		}
		\begin{tablenotes}
			\item[*] The selected graph is TW on 9-machine cluster in the experiment section.
		\end{tablenotes}
	\end{threeparttable}
\end{table}
\vspace{-0.1in}
% heterogeneous machine config is also in the experiment section

\subsection{Related Work} \label{sec:related}

Graph partition problem includes offline partition and streaming partition: offline partition reads the graph data entirely into memory and divide the graph structure; in contrast, streaming partition reads the graph data batch by batch, and it needs to decide the location of nodes/edges in each batch immediately.
In this paper, we mainly focus on offline partition.
Existing work related to offline graph partition can be mainly divided into two categories: vertex-centric partition (a.k.a., edge-cut) and edge-centric partition(a.k.a., vertex-cut).

\Paragraph{Vertex-centric Partition}.
The edge-cut solutions divide vertices of graph $G$ into different partitions, ensuring the load balancing and minimizing the edge-cut.
This can be formulized as $max_{i\in [i,p]}{|V_{i}|}\leq \frac{\alpha^{\prime}|V|}{p}$ and $min\{\sum_{\overline{uv}\in G}{S(u)\neq S(v)}\}$.
Some earlier solutions uses random hash to assign a partition for each vertex $v$, e.g., $f(v)=hash(v)\% p$.
Though the random hash is very fast, it destroys the graph locality, thus has high edge-cut value.
In contrast, METIS \cite{METIS} adopts a multi-level paradigm, which makes better use of the locality.
However, it is rather slow and has prohibitive memory occupation on large graphs.
Later, LDG \cite{LDG} and Fennel \cite{Fennel} utilize graph locality in a simpler way: put adjacent nodes together in one partition to reduce edge-cut.
Greedily, they assign the vertex $u$ to the partition that has the most neighbors of $u$.
Overall, on scale-free graphs whose degrees follow power-law distribution, some super nodes have $>10^6$ edges, which can cause severe bottleneck in a single machine due to load imbalance and terrible communication.
Therefore, edge-centric partition is preferred on large power-law graphs.

\Paragraph{Edge-centric Partition}.
The vertex-cut solutions divides edges of graph $G$ into different partitions, ensuring the load balance and minimizing the vertex-cut metric (i.e., replication factor).
This can be formulized as $max_{i\in [i,p]}{|E_{i}|}\leq \frac{\alpha^{\prime}|E|}{p}$ and $min\{RF\}$.
Some earlier solutions (e.g., PowerGraph \cite{PowerGraph}) intuitively use greedy algorithms, prioritizing the partitions that have two endpoints of the given edge.
In order to reduce $RF$, \cite{PowerGraph} limits the number of replication for each vertex to $2\sqrt{p}$ in the algorithm design.
DBH \cite{DBH} and Ginger \cite{PowerLyra} utilize the power-law distribution and ensure that edges of low-degree vertices are put together.
More advanced HDRF \cite{HDRF} computes a score for each partition when assigning an edge $e$ and choose the partition with the highest score.
The main difference is that for each edge $\overline{uv}$, HDRF selects the partition that $u$ or $v$ has the largest partial degree, while DBH selects the partition that $u$ (assuming $deg(u)<deg(v)$) resides in.
The state-of-the-art algorithms are NE \cite{NE} and EBV \cite{EBV}.
By linear graph exploration, NE generates the partitions one by one and achieves the lowest $RF$ currently.
It improves the locality of each subgraph (partition $G_{i}$) by choosing node $u$ that have the lowest $|N(u)\setminus V_{i}|$ during exploration.
However, NE has two main shortcomings, which limit its performance.
On the one hand, NE neglects nodes that have large $|N(u)\bigcap V_{i}|$ but $|N(u)\setminus V_{i}|$ is not the minimum.
Thus, it selects low-deg nodes frequently, which can generate a long-tail DFS search path with tiny cohesion.
For example, in Figure \ref{fig:ne-dfs}, NE chooses the sequence (`A',`B',...,`F') rather than the node `G' that has more connection with $V_i$.
On the other hand, NE may stop extending the current partition on high-degree nodes, causing high communication cost.
The latest algorithm EBV quantifies the proportion of replicated vertices and the balance of vertices/edges assignment as significant parameters.
It also sorts the order of edges by the sum of end-vertices' degrees from small to large.
In this way, it can accelerate the processing of power-law graphs by alleviating the imbalance of communication between multiple machines.

\vspace{-0.15in}
\begin{figure}[htbp]
	\centering
	\includegraphics[width=7cm]  {\picfolder 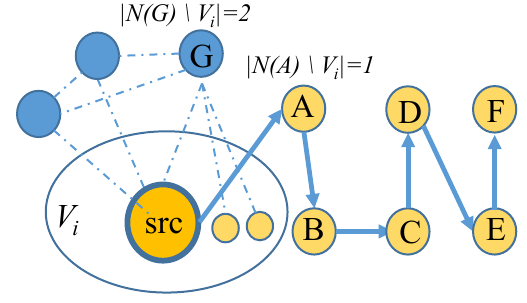}        	
	\vspace{-0.15in}     	
	\caption{The shortcoming of NE}      	
	\label{fig:ne-dfs}  
\end{figure}
\vspace{-0.15in}

Both of the two kinds of solutions above can only support homogeneous machines.
They can not guarantee on generating feasible partitions on heterogeneous machines and the quality of partition is unacceptable even if some solution exists.
Besides, existing solutions mainly optimize for power-law graphs, but there are still some real-life graphs that are not scale-free, e.g., mesh-like graphs \emph{road}.

\Paragraph{Heterogeneous Distributed Computing}.
Distributed computing on big data has been well studied \cite{SurDataPart,SurDataAna}, but in graph computing we need to utilize the characteristics of both graphs and clusters.
Recently, the heterogeneous computing on deep neural network (DNN) arises \cite{HeterPS,Whale}, which leads the trend of enabling heterogeneous resources.
The literature includes two works for heterogeneous graph algorithms: \cite{HeterCompPart} only consider different computing power, while GrapH \cite{GrapH} targets at various communication cost.
\cite{HeterCompPart} coarsen $G$ first, then partition it and finally project it back to $G$.
GrapH uses streaming partition and group machines into different clusters according to their network traffic.
Due to the lack of collaborative optimization on machines with heterogeneous memory, computing power and network traffic, none of them can achieve a good balance between computation and communication.
They use different framework and can not be combined with other solutions directly, thus we compare WindGP with them in Section \ref{sec:distributed}.
HaSGP \cite{HaSGP} considers the heterogeneity of both computing and communication, but it has three limitations: (1) neglects different memory capacity; (2) as a streaming algorithm, lacks  optimization in subgraph locality; (3) targets at high-bandwidth network and considers multi-core conflict.
In its experiments, only four and 32 nodes are used and the heterogeneity of both computing and communication are not considered in the same cluster.
Besides, heterogeneous computing is mainly used in low-end scenarios whose network bandwidth is far smaller than memory bandwidth, thus limiting the application of HaSGP.
The state-of-the-art HAEP \cite{HAEP} adopts the same metrics (balance ratio $\alpha^{\prime}$ and replication factor $RF$) as homogeneous cases, and proposes heuristic neighbor expansion to improve subgraph locality.  
However, $\alpha^{\prime}$ and $RF$ can not depict the quality of heterogeneous partition well, as analyzed in Section \ref{sec:problem}.
In addition, HAEP includes both computing and communication heterogeneity in the same cluster, but still omits the memory heterogeneity.
Furthermore, the maximum cluster and the maximum dataset only contain 32 nodes and 117M edges respectively, which is not enough to prove the scalability.  
%For fair comparison, we combine their partition strategies together and apply to existing edge partition methods like NE.
%can not be easily combined

%\nop{
\Paragraph{Variants}.
For large graphs, the memory of a single machine is not enough, thus we need streaming partition or distributed partition. 
In some scenarios, the graph data come as time series, this can only be handled by streaming algorithms \cite{DBLP:journals/pvldb/AbbasKCV18,DBLP:conf/sigmod/PacaciO19}.
Generally, hash-based methods perform terrible on communication cost, but they can adapt to streaming graph partition naturally.
With the rising of AI techniques, learning-based partition also develops a lot.
For example, GAP \cite{GAP} is an approximated edge-cut algorithm based on graph neural network.
%}

%\nop{
\Paragraph{Optimization Methods}.
Heterogeneous-machine graph partition problem is closely related to \emph{combinatorial optimization}. 
Exact methods like branch and bound \cite{lawler1966branch} can find exact solution for small-size problem. 
Open-source and commercial solvers like COIN-OR CBC \cite{forrest2005cbc}, ZIB SCIP \cite{gamrath2020scip} and GUROBI \cite{bixby2007gurobi} have integrated powerful heuristics and can solve medium problems within acceptable time. 
However, in realistic applications that have billions of edges and millions of nodes, exact methods can hardly find a feasible solution due to the exponential complexity of computation. 
Thus, approximation methods and local search methods \cite{hromkovivc2013algorithmics} and \cite{pisinger2010large} have been proposed. 
%}

%\nop{
\begin{table}[htbp]
\small
\centering
\caption{Notations}
\vspace{-0.1in}        
\begin{tabular}{|c||p{5cm}|}
    \hline
    $G$& graph to be partitioned \\
    \hline
    $v,u$& vertices in a graph or subgraph \\
    \hline
    $G_{i}$& the $i$-th partition of $G$ \\
    \hline
    $V_{i},E_{i}$& the node set and edge set of $G_{i}$ \\
    \hline
    $M^{node},M^{edge}$& the memory occupation of a node and an edge, respectively\\
    \hline
    $Machine_{i}$& the $i$-th machine \\
    \hline
    $M_{i},C_{i}^{node},C_{i}^{edge},C_{i}^{com}$& the memory size, the computing cost of a node, the computing cost of an edge and the communication cost of a node in the $i$-th machine, respectively\\
    \hline
    %$c_{v}^{ij}$& whether $v$ exists in both $i$-th and $j$-th machines \\
    %\hline
    $T_i, T_{i}^{cal}, T_{i}^{com}$& the total cost, computing cost and the communication cost the $i$-th machine, respectively \\
    \hline
    $S(u)$& the set of partitions that $u$ exists \\
    \hline
    $N(u)$& the neighbor set of vertex $u$ in $G$ \\
    \hline
    $N_{i}(u)$& the neighbor set of vertex $u$ in $G_{i}$ \\
    \hline
    $deg(u)$& the degree of vertex $u$ in $G$ \\
    \hline
    $deg_{i}(u)$& the degree of vertex $u$ in $G_{i}$ \\
    \hline
    $n_{i,j}$ & the number of replica nodes between partition $i$ and $j$  \\
    \hline
    $|L|$ & The size of set $L$ \\
    \hline
    $num(L)$ & the number of currently valid elements in set $L$   \\
    \hline
    $\delta_i$ & the precomputed capacity of partition $G_i$   \\
    \hline
    $I_B(v)$ & the indicator function: whether $v$ exists in other partition   \\
    \hline
    $\alpha^{\prime}, RF$ & the load balance ratio and the replication factor in homogeneous partitioning, respectively   \\
    \hline
    $\alpha$ & the threshold controlling the balance between $N(u)\setminus S$ and $N(u)\cap S$ \\
        \hline
    $\beta$ & the threshold controlling the impact of border vertices  \\
        \hline
    $\gamma$ &  the threshold of edges above which partitions should be destroyed   \\
        \hline
    $\theta$  & the proportion of edges to be removed in a partition   \\
        \hline
    $N_0, T_0$  & the number of local and global try times    \\
        \hline
\end{tabular}
\label{tab:notations}
\end{table}
%}

\section{Algorithm}\label{sec:algorithm}

\subsection{Overview}

%BBNE(Border Balanced Neighbor Expansion),

In this section, we propose a comprehensive and general edge partitioning algorithm WindGP, which can divide the original graph $G$ to $p$ subgraphs $G_i=(V_i,E_i)$ ($p$ subgraphs corresponding to $p$ machines), where $E_i \cap E_j=\emptyset, \forall i\neq j$  with low communication cost and balanced calculation cost. 
This algorithm consists of three main parts. 
Firstly, the upper bound ($\delta_i$) of the edge set ($|E_i|$) is obtained by \emph{graph-oriented preprocessing}, which can provide a good guidance for the partition expansion in the second part.
Secondly, for each partition $E_i$, edges are iteratively added to $E_i$ by \emph{best-first search} until $|E_i|\geq \delta_i$. 
Thirdly, our post-processing performs \emph{subgraph-local search} on the initial solution $\{E_i|1\leq i\leq p\}$ obtained in the previous two steps to generate a better solution $\{E_i^*\}$ with lower total cost $TC$ (see Definition \ref{def:problem}). 
Note that the previous two steps can generate a good solution with higher quality than traditional algorithms, which means the final local search may not be performed if the running time is limited (e.g., in real-time graph processing scenarios).
The overview of our framework is in Figure \ref{fig:frame}.

\vspace{-0.15in}
\begin{figure}[H]
	\centering
	\includegraphics[width=8cm]{\picfolder 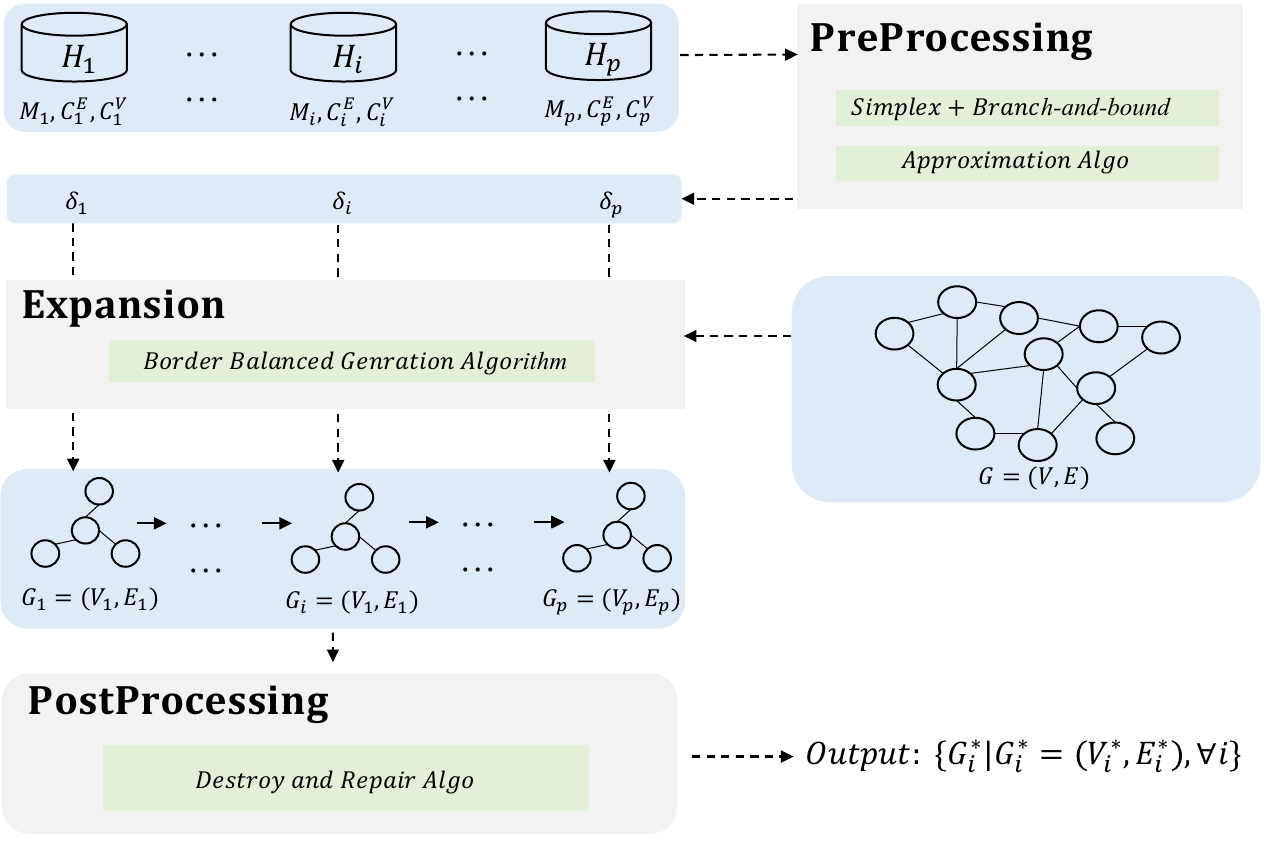}
    \vspace{-0.1in}
	\caption{The framework of WindGP}
	\label{fig:frame}
\end{figure}
\vspace{-0.15in}

\subsection{Graph-oriented Preprocessing}\label{sec:preprocess}

The first part of our algorithm is to compute the upper bound ($|\delta_i|$) for each subgraph $G_{i}$, $\forall i\in [1,p]$. 
Recall that the objective function in the problem statement tries to minimize the maximum of the sum of calculation time and communication time. 
Obviously, the optimal $\delta_i$ corresponds to the optimal partition and the computation of optimal $\delta_i$ is also NP-hard.
Thus, we convert the problem to a lightweight \emph{mixed integer programming} (MIP) problem by utilizing graph characteristics. 
Note that calculation time of each machine is solely determined by the number of edges and vertices of $G_i$. 
In contrast, the communication time is much complicated and much harder to be balanced.
Experiments on graphs with hundreds of edges show that the calculation time is nearly balanced in the optimal solution, while the communication time is not.
This inspires us to balance the calculation time first, which results in a near-optimal solution space.
Thus, we need to assign appropriate upper bound for each subgraph to achieve the balance.
Further improvement can be achieved by the final post-processing, which may slightly disturb the balance of calculation to get better solutions.
In this way, the computation of optimal $\delta_{i}$ can be formulated below: 
\begin{equation}
    \small
	\label{eq:compute}
	\begin{aligned}
		&\text{minimize} & &\lambda &   \\
		&\text{subject to}& & \sum_{i} |E_i|=|E| &   \\
		& & & |E_i|\leq \dfrac{M_i}{M^{edge}+M^{node}\times |V_{i}|/|E_{i}|} & \forall 1 \leq i\leq p \\
 		& & & \lambda \geq C_i^{edge} |E_i|+C_i^{node} |V_i|& \forall 1 \leq i\leq p \\
		& & & |E_i|\in \{1,2,\ldots,|E|\} &\forall 1 \leq i\leq p \\
	\end{aligned}
\end{equation}
where $(|E_i|,|V_i|,|E|,|V|)$ means the number of elements of $(E_i, V_i, E, V)$, respectively. 
Here auxiliary decision variable $\lambda$ represents the maximum calculation time of each machine and our objective is to minimize $\lambda$. 
The first constraint shows that the sum of these capacities is equal to the edge size of the original graph $G$. 
In addition, the memory occupation in each subgraph should not exceed the memory size of the corresponding machine, as shown in the second constraint.
Furthermore, the third constraint requires that the maximum calculation time is not smaller than the calculation time of each machine.
Finally, the fourth constraint indicates that all capacities should be integers within $[1,|E|]$.

To further simplify the problem, the vertex set size of machine $i$ (i.e., $|V_i|$) can be replaced by the average ratio in the entire graph as $|V_i|=\frac{|V|}{|E|}\times |E_i|$.
This simplification works because each partition is expected to be a ``normal'' graph which has characteristics similar to the original graph $G$.
In practice, the number of edges is much larger than that of vertices even on sparse graphs (usually $10\sim 100\times$), not to mention dense graphs.
Besides, the computation of a vertex is usually less costly than the computation of an edge.
Therefore, this simplification is error-bounded, and in all experiments it does not affect the search of optimal solutions.
After the simplification, let $C_i=C_i^{edge}+\frac{|V|}{|E|}\times C_i^{node}$, the original equations are equivalent to  
\begin{equation}
    \small
	\label{eq:simplex}
	\begin{aligned}
		&\text{minimize} & &\lambda &   \\
		&\text{subject to}& & \sum_{i} |E_i|=|E| &   \\
		& & &(M^{edge}+M^{node}\times|V|/|E|)|E_i|+\alpha_i=M_i  \\
		& & & C_i |E_i| -\lambda + \beta_i =0  \\
		& & & |E_i|\in \{1,2,\ldots,|E|\}  \\
         %& & & \alpha_i,\beta_i \geq 0  \\
	\end{aligned}
\end{equation}
where $\{\alpha_i$,$\beta_i \geq 0|1\leq i \leq p\}$ are auxiliary variables used to convert inequalities into equations. 
The above MIP problem has $p+1$ decision variables and $2p+1$ constraints.

\subsubsection{Exact Method}
General integer programming (IP) problems have been proven to be NP-hard \cite{conforti2014integer}. 
As a generalization of IP, MIP problems are also NP-hard and viewed as one of the most challenging areas in applied mathematics.  
Most of the state-of-the-art MIP solvers (e.g., ZIB Scip \cite{gamrath2020scip} and Gurobi \cite{bixby2007gurobi}) implements a tree search algorithm framework called branch-and-bound \cite{lawler1966branch}, and they can solve the problem efficiently if the machine number is not too large. 
However, as the data volume and the computing power grow exponentially \cite{denning2016exponential}, there are many application scenarios where the problem becomes too large to process for a MIP solver. 
To resolve this issue, we propose a graph-oriented heuristic that can help solve the MIP problem iteratively.

\begin{algorithm}
	\small
	\caption{Computing the bound $\delta_i$ for edge partitions}
	\label{alg:preprocess}
	\KwIn{$|E|, C_i^{edge}, M_i, \forall i \in[1,p]$}
	\KwOut{$\delta_i, \forall i\in [1,p]$}
	\textbf{procedure} $GeneratingCapacity(|E|,C_i^{edge},M_i)$  \\
	$R=|E|, I=\{1,2,...,p\}$ \\ 
	\While{$R > 0 $}
	{	
		$T = \sum_{i\in I} \frac{1}{C_i}$  \\
		\For{$i \in [1,p]$}{ \label{algcmd:preFor1}
			\If{$\delta_i$ already allocated}
			{
				continue   \\
			}
			$\delta_i^1 = \frac{R}{T}\times \frac{1}{C_i}$ \\
			$\delta_i^2 = \frac{|M_i|}{M^{edge}+M^{node}\times|V|/|E|}$ \\
			%$\delta_i^{*} = min\{\delta_i^1, \delta_i^2\}$ \\
			\If{$\delta_i^{1} > \delta_{i}^2$} 
			{
				$\delta_i = \delta_i^{2}$, $R \gets R-\delta_i$, $I\gets I\setminus \{i\}$  \\
			}
		}
\nop{
		\If{no new $\delta_i$ allocated} 
		{ \label{algcmd:allMeet}
			\For{$i \in [1,p]$}{  \label{algcmd:preFor2}
				\If{$\delta_i$ already allocated}
				{
					continue   \\
				}
				$\delta_i = \frac{R}{\sum_{i\in I} \frac{1}{C_i}}\times \frac{1}{C_i}$ \\
			}
			break  \\
		}
}
	}
	\Return $\delta_i, \forall i\in [1,p]$  \\
\end{algorithm}

\subsubsection{Graph-oriented Heuristic}
The intuition behind the heuristic is straightforward. 
If we omit the heterogeneous memory sizes of machines, in the best case the calculation time should be the same for each machine. 
Based on this claim, we try to allocate edges to each machine so that their computation time is a constant (i.e., let $C_i\times \delta_i=\omega, \forall i\in [1,p]$). 
Thus, we have $|E|=\sum_{i}{\delta_i}=\omega \sum_{i} \frac{1}{C_i}$, and $\delta_i$ can be calculated by $\frac{|E|}{\sum_{i} \frac{1}{C_i}}\times \frac{1}{C_i}$.
However, the computed edge size may exceed the memory size of the machine. 
In such cases, we choose to limit its edge size with respect to the memory size of the corresponding machine. 

Algorithm \ref{alg:preprocess} lists the details of our heuristic.
The algorithm first estimates the capacity of the edge set for each machine, i.e., $\delta_i^1$.
Then it checks whether $\delta_i^1$ meets the constraint of memory size.
If the memory consumption exceeds the memory size, we fix its edge set capacity so that both the memory size constraint and the integer constraint are respected. 
%In each round, if all $\delta_i^1$ meets restrictions (Line \ref{algcmd:allMeet}), all unallocated machines are allocated directly within this round.
The same process is repeated recursively for the remaining machines and remaining edge number until no edge is left. 

\begin{lemma}\label{lma:preprocess} 
If not consider the integer constraints, Algorithm \ref{alg:preprocess} can find the optimal solution.
\end{lemma}
\begin{proof} \label{prf:lemma1}
Apparently, our solution can find the optimal if all machine's memory is enough (Line 10 never occurs).
Let $f(E,p)=\lambda$, we have $f(E,p)\geq f(E^{\prime},p)$ when $E>E^{\prime}$.
By mathematical induction, we assume $f(E,p)$ is optimal $\forall E^{\prime}<E$.
For specific $E$ and $p$, assume $i$ is the machine that satisfies Line 10, then all of $M_i$ is used by Algorithm \ref{alg:preprocess} and $f(E_i,1)<f(E-E_i,p-1)$.
If this is not optimal, it must be $E^*_i<E_i$ in the optimal solution $f^*(E,p)$, thus $f(E-E_i,p-1)<f^*(E-E^*_i,p-1)$, which violates the optimal condition.
\end{proof}
\begin{theorem}\label{trm:preprocess} \textbf{Approximation Error Bound}
    Compared with the optimal solution for Equation \ref{eq:simplex}, the error bound of Algorithm \ref{alg:preprocess} is $\frac{p^2}{|E|}$.
\end{theorem}
\begin{proof} \label{prf:preprocess}
If $\delta^1_i$ is not integer and is smaller than $\delta^2_i$, we always use the flooring integer.
If this not optimal, we can compare the difference with the optimal (using the ceiling integer).
The difference part is only one edge, thus we have $f(E+1,p)-f(E,p)\leq C_i$ and $f(E,p)-f^*(E,p)\leq pC_i$ as in each iteration at least one partition can be decided.
$f^*(E,p)\geq \frac{|E|C_i}{p}$, thus the error is bounded by $\frac{pC_i}{\frac{|E|C_i}{p}}=\frac{p^2}{|E|}$.
\end{proof}
%if the threshold is not integer, maybe tiny error exists, but it can be bounded.
In general settings ($|E|>10^7$), the error bound is much smaller than $\frac{1}{10}$.
Experiments in Section \ref{sec:evaluate} also verifies the quality of Algorithm \ref{alg:preprocess}.
%experiment (10% diff? similar perf?), Approximate bound

\Paragraph{Analysis}.
In each round, at least one machine is allocated or all remaining machines are allocated.
Thus, the number of rounds is bounded by $p$.
Within each round, two $for$ loops are needed excluding the variables that can be pre-computed (e.g., $C_i$, $\sum_{i}{\frac{1}{C_i}}$ and $\delta_i^2$).
Overall, the time complexity  of Algorithm \ref{alg:preprocess} is $O(p\times 2p+p)=O(p^2)$, which is far more efficient than the exponential-complexity MIP solvers.
This makes the preprocessing the least time-consuming part in the entire WindGP algorithm as $p$ (usually $<10^3$) is much smaller than the size of graph $G$ ($\sim 10^7$).
The space complexity is linear ($O(p)$), which can be deduced simply.

\subsection{Partition Expansion: Best-first Search}\label{sec:best-search}

Once the edge set capacities are computed for each machine, we can generate $p$ partitions for these machines one by one.
%Similar to the previous NE algorithm \cite{NE}, we partition the graph iteratively for each machine. 
Specifically, for machine $i$, edge set $E_i$ is selected from the working partition $G_i$ containing all unassigned edges so far.
The number of edges in each partition is strictly restricted by the capacity, thus the computation cost is marginally balanced on heterogeneous machines.
The next challenge is the balance of communication cost, which is the goal of this section.
Let \emph{core set} $C$ be the set of vertices whose unassigned edges are all allocated successfully in current machine and \emph{boundary set} $S$ be the vertex set covered by $E_i$.   
During each iteration, one vertex $v_{sel}$ is selected for expansion according to the strategy below. 

%$\bullet$\textbf{Degree Balanced Generation}.
\Paragraph{Degree Balanced Generation}.
Heuristically, vertices with lower degree and shorter distance (from new candidate vertices to the core set $C$) are preferred.
If $S\setminus C=\emptyset$, $v_{sel}$ is selected by the \emph{vertexSelection} procedure, which can be designed from perspective of \emph{degree} and \emph{distance} instead of naive random selection. 
Otherwise, in order to improve the subgraph cohesiveness, we focus on both two kinds of edges ($N(u)\setminus S$ and $N(u)\cap S$), while NE only considers the first kind.
Considering two nodes $u$ and $v$ such that $|N(u)\setminus S|=|N(v)\setminus S|$ and $|N(u)\cap S|>|N(v)\cap S|$, $u$ is preferred because it contributes to higher cohesiveness. 
A hyper-parameter $\alpha (0\leq \alpha\leq 1)$ is built to control the balance:
\begin{equation}
    \small
	\begin{aligned}
        v_{sel} & =   argmin_{v\in S\setminus C}[|N(v)\setminus S|-\alpha|N(v)\cap S|]  \\
    & = argmin_{v\in S\setminus C}[|N(v)\setminus S|-\alpha(|N(v)|-|N(v)\setminus S|)]  \\
    & = argmin_{v\in S\setminus C}[(1+\alpha)|N(v)\setminus S|-\alpha|N(v)|]  \\
	\end{aligned}
\end{equation}
In practice, $\alpha$ can be fine tuned to improve the performance.

%$\bullet$\textbf{Border Generation}.
\Paragraph{Border Generation}.
In order to reduce the communication cost, we pay more attention to the \emph{border vertex} (i.e., vertices that exist in multiple machines). 
Obviously, for a border vertex $v$, its communication cost is smaller if $v$ only exists in two machines. 
Let $B$ bet the set of all border vertices that have been found. 
Namely, $B$ should be updated with all border vertices of partition $i$ after the expansion procedure of partition $i$ finishes. 
Let $I_B$ be the indicator function as follow:
\begin{equation}
    \small
	I_B(v)=
	\left\{
	\begin{aligned}
		1&  , & \mbox{if } v \in B \\
		0&  , &  \mbox{if } v \notin B \\
	\end{aligned}
	\right.
\end{equation}
During the expansion, we prefer the border vertices controlled by a hyper-parameter $\beta$ ($0\leq \beta \leq 1$). 
Combined with the Degree Balanced Generation strategy, we define the final priority function as follow:
\begin{equation}
    \small
	w(v) = (1+\alpha)|N(v)\setminus S|-(\alpha+I_B(v)\beta)|N(v)|
\end{equation}
And,
\begin{equation}
    \small
	v_{sel} = argmin_{v\in S\setminus C}w(v)
\end{equation}
The tuning of $\alpha$ and $\beta$ is in our experiments.

The combination of degree balanced generation and border generation is called the \emph{Best-First Search}, which is superior to traditional breadth-first search (BFS) or depth-first search (DFS).
During the graph exploration, BFS maintains a frontier queue and expand these frontiers one by one, then it inserts new frontiers to the queue.
In contrast, best-first search expands from one of the frontiers that have the minimum $w(v)$ and update the frontier queue.
Our method is better than NE expansion \cite{NE}, because NE only consider the size of $N(v)\setminus S$ and may generate long-tail DFS paths (as analyzed in Section \ref{sec:related}).

Taking Figure \ref{fig:expansion} as a specific example, let $X$, $Y$ and $Z$ be the candidate vertices and $X$ be the node that exist in previous partitions, i.e., $X\in B$. 
The thick edges are the ``incoming'' edges, namely the edges connecting to set $S$, and the thin edges are the ``outgoing'' edges. 
Assuming that  $\alpha=0.3$ and $\beta=0.3$, we have $w(X)=(0.3+1)*2-(0.3+0.3*1)*4=0.2$, $w(Y)=(0.3+1)*2-(0.3+0)*5=1.1$, $w(Z)=(0.3+1)*2-(0.3+0)*4=1.4$.
Consequently, the vertex $X$ should be expanded in current step.

    \vspace{-0.15in}
\begin{figure}[H]
	\centering
	\includegraphics[width=7cm]{\picfolder 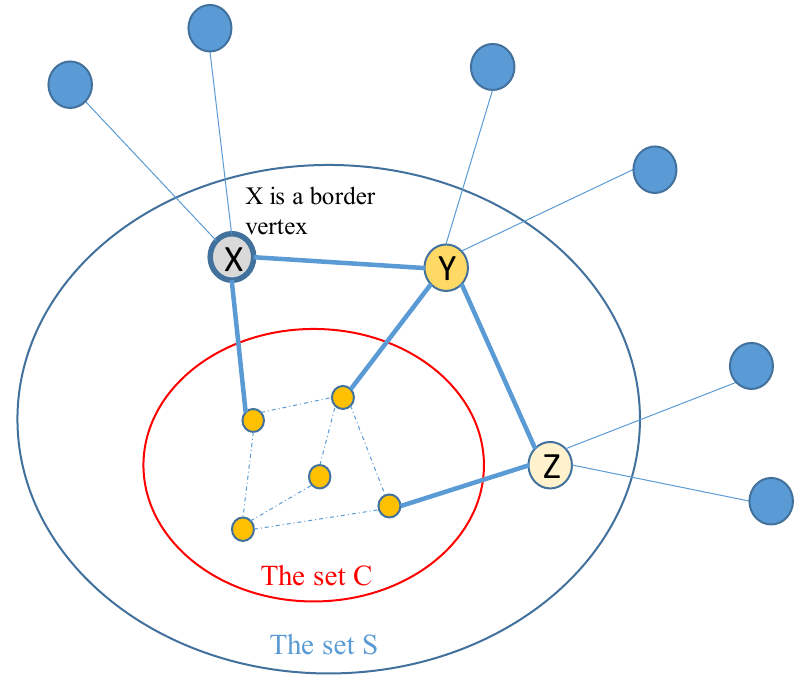}
    \vspace{-0.15in}
	\caption{An example of vertex expansion}
	\label{fig:expansion}
\end{figure}
    \vspace{-0.15in}

The entire procedure of our heuristic expansion is shown in Algorithm \ref{alg:expand}, while the update rule of core set and boundary set is described in Algorithm \ref{alg:alloc}. 
To generate one edge partition $E_i$, WindGP expands from a selected vertex iteratively until the size of this partition exceeds the capacity $\delta_i$ (computed in Section \ref{sec:preprocess}).
During each iteration, the vertex $v$ with the minimum $w(v)$ is selected if the candidate set is not empty; otherwise, the vertex is chosen from all remaining vertices in $V$.
Once $v$ is decided, $AllocEdges$ procedure (Algorithm \ref{alg:alloc}) is called to expand $v$ and process the 1-hop and 2-hop neighbors of $v$.
Note that $D=N(x)\setminus S$ are new boundary nodes, and the edges between $D$ and previous boundary set are moved from $E$ to $E_i$.
These new added edges belongs to the remaining edge set $E$ rather than $E(G)$ and they also include $\overline{xy}$ (Line \ref{algcmd:addEdge} of Algorithm \ref{alg:alloc}).
Finally, the border vertex set $B$ is supplemented with the new generated border vertices in current partition.

\Paragraph{Analysis}.
All edges are allocated in Algorithm \ref{alg:alloc}, thus the total time complexity of calling $AllocEdges$ in Algorithm \ref{alg:expand} is $O(|E_i|)$ as each edge is only removed once from $E$.
In practice, the operation of vertex selection in Line \ref{algcmd:sel1} and Line \ref{algcmd:sel2} can be accelerated by \emph{Min-Heap} \cite{MinHeap}, and the heap size is bounded by $V_i$.
Therefore, the total complexity of these two lines is $|V_i|\log{|V_i|}$ as the \emph{While} loop is executed at most $V_i$ times (i.e., each vertex only enters the heap once).
Besides, the set operations (including set union, set intersection and set minus) can be implemented by \emph{bitmap} \cite{SQLG+}, which can optimize the complexity of Line \ref{algcmd:border} to $O(|S\setminus C|)=O(|V_i|)$.
To sum up, the total time complexity of Algorithm \ref{alg:expand} is $O(|E_i|+|V_i|+|V_i|\log{|V_i|})=O(|E_i|+|V_i|\log{|V_i|})$.
As for space complexity, the set structure $C$ and $S$ contain no more than $V_i$ nodes, and the set structure $B$ contain no more than $V(G)$ nodes, thus the space complexity is $O(|V_i|+|V(G)|+|E_i|)=O(|V(G)|+|E_i|)$ except for the storage of the entire graph $G$.

\begin{algorithm}
	\small
	\caption{Generate one edge partition $E_i$}
	\label{alg:expand}
    \KwIn{$E(G)\setminus \sum_{j<i} E_j, V, \delta_i, \alpha, \beta, B$}
	\KwOut{$E_i$}
	\textbf{procedure} $EXPAND(E(G)\setminus \sum_{j<i} E_j,\delta_i)$  \\
	$C,S,E_i \gets \emptyset$  \\
	\While{$|E_i| \leq \delta_i$}
	{
		\If{$S\setminus C = \emptyset$} 
		{
            $x \gets vertexSelection(V\setminus C)$ \label{algcmd:sel1} \\
		}
		\Else
		{
            $x \gets argmin_{v\in S\setminus C}[(1+\alpha)|N(v)\setminus S|-(\alpha+I_B(v)\beta)|N(v)|]$  \label{algcmd:sel2}  \\
		} 
		$AllocEdges(C,S,E_i,x, E\setminus \sum_{j<i} E_j)$
	}
    $B \gets B\cup(S\setminus C)$ \label{algcmd:border}  \\
    \Return $E_i$ \\
\end{algorithm}

\begin{algorithm}
	\small
	\caption{Allocate edges for core vertex $x$}
	\label{alg:alloc}
	\textbf{procedure} $AllocEdges(C,S,E_i,x, E)$  \\
	$C\gets C\cup\{x\}, S\gets S\cup\{x\}$  \\
	\ForEach{$y\in N(x)\setminus S$} 
	{
		$S\gets S\cup\{y\}$   \\
		\ForEach{$z\in N(y)\cap S$}
		{
            $E_i\gets E_i\cup \{\overline{yz}\}$ \label{algcmd:addEdge}   \\
			$E \gets E\setminus \{\overline{yz}\}$  \\
			\If{$|E_i|\geq \delta_i$}
			{
				\Return
			} 	
		}	
	}
\end{algorithm}

\subsection{Post-Processing: Subgraph-local Search}\label{sec:post-process}

%\Paragraph{Large Neighborhood Search}.
In this section, Subgraph-Local Search (SLS) is proposed to improve the partition generated by previous stages. 
According to \cite{hromkovivc2013algorithmics,pisinger2010large}, the quality of edge partition result can be enhanced by moving or swapping edges between partitions.
By utilizing the characteristics of local subgraph, SLS iteratively finds a better solution in the neighborhood of current solution. 
The operator of designing a neighborhood in the solution space and finding better solution by SLS is critically important. 
A ``good'' operator should not only specify the promising neighborhood that leads to better solutions, but also avoid getting stuck in local optimal.  
In this paper, we design two SLS operators. 
The first one is the \emph{destroy-and-repair} operator that aims to find better solutions, and the second one is the \emph{re-partition} operator that attempts to get escaped from local optima.
We first apply the first operator to current result and count the number of consecutive \emph{fail-to-improve} times.
If it exceeds the pre-determined threshold ($N_0$, set to 5 by default), the current result is viewed as local optimal and the re-partition operator is applied. 
Algorithm \ref{alg:SLS} gives the main procedure.

\begin{algorithm}
	\small
	\caption{Main framework of SLS}
    \label{alg:SLS}
	\textbf{procedure} SLS($\{E_i|1\leq i\leq p\}, T_0, N_0,k$)  \\
	/* $N_0$ and $T_0$ are the number of local and global try times, respectively  */    \\
	\While{$T_0 > 0$}
	{
		\If{$DestroyRepair(\{E_i|1\leq i\leq p\})$}
		{
			$n \gets 0$ \\
		}
		\Else
		{
			$n \gets n+1$    \\
			\If{$n \geq N_0$}
			{
				$REPARTITION(\{E_i|1\leq i\leq p\},k)$  \\
				$n \gets 0$  \\
			}
		}
		$T_0 \gets T_0 - 1$  \\
	}
\end{algorithm}

%In the SLS algorithm, neighborhood is defined by the \emph{DestroyRepair} procedure which consists of destroyer and repairer. 
%During the SLS processing, some parts of the current solution will be temporally removed by the destroyer, and another complete solution will be rebuilt by the repairer. 
\Paragraph{Destroy-and-Repair}.
When applying the destroy-and-repair operation, part of the current partition is removed by a destroy step and  rebuilt by a repair step while the remaining  part is preserved.
% The purpose of destroy is to identify the bad parts of the current solution for future re-construction while preserving the good part of the current solution. The destroy operator in our definition of the Graph Edge Partition Problem primarily consists of two choices: the first choice is which machines to destroy, and the second choice is which edges to destroy when machines are provided. We can choose to destroy the machine with longer communication time, the machine with longer computation time, or both since both computation time and communication time are taken into account by the objective function. In practical algorithm design, we determine a threshold value based on the following formula:
% Minimum communication time (computation time) + ratio * (Maximum communication time (computation time) - Minimum communication time (computation time))
% and select the machines with communication time (computation time) greater than this threshold value. When destroyed machines are chosen, the final $10\%$ of the edges added to  these machines during the construction of the initial solution are eliminated, then we can obtain a partial solution with a strictly superior objective function.
In the destroy step, the non-optimal part is specified for future re-construction. 
Hence there are two key decisions: the first one is which machines to destroy, and the second one is which edges to destroy in the selected machine. 
A quantile parameter $\gamma$ is leveraged to decide the threshold (Line 4 of Algorithm \ref{alg:destroy}), and machines with time cost above this threshold are destroyed. 
For each destroyed machine, a proportion $\theta$ ($0<\theta<1$, set to 1\% by default) of total edges are removed. %The edge remove order is the inverse order of edge insertion. 
The \emph{last-in-first-out} rule is adopted to select the edges to be removed, because in this way the connectivity of edges can be preserved in current machine. 
By adjusting $\gamma$, we can control the threshold, which is $min\{T_i\}$ and $max\{T_i\}$ when $\gamma$ is 0 and 1 respectively.
In our work, $\gamma$ is set to 0.9 by default, which is easier for repairing.
%We can obtain a partial solution with a strictly superior objective function using destroyer based on communication time of each machine.

%The purpose of destroy is to identify the bad parts of the current solution for future re-construction while preserving the good part of the current solution. 
%The destroy process primarily consists of two choices: the first choice is which partition to destroy, and the second choice is which edges to destroy when the partition is selected.
%We can obtain a partial solution with a strict superior objective function using destroyer based on communication cost of each machine.

%By reassigning destroyed edges, the repair process finds a solution that performs better than the current solution. 
%It is necessary to decide which machine is the most suitable one for each destroyed edge before placing it there. 
%We use greedy heuristic assignment methods to balance the communication cost in repair procedure. 
%In most cases, it can improve the maximum communication time among all machines at the cost of total communication time.

    \vspace{-0.15in}
\begin{figure}[htbp]
	\centering
	\includegraphics[width=7cm]{\picfolder 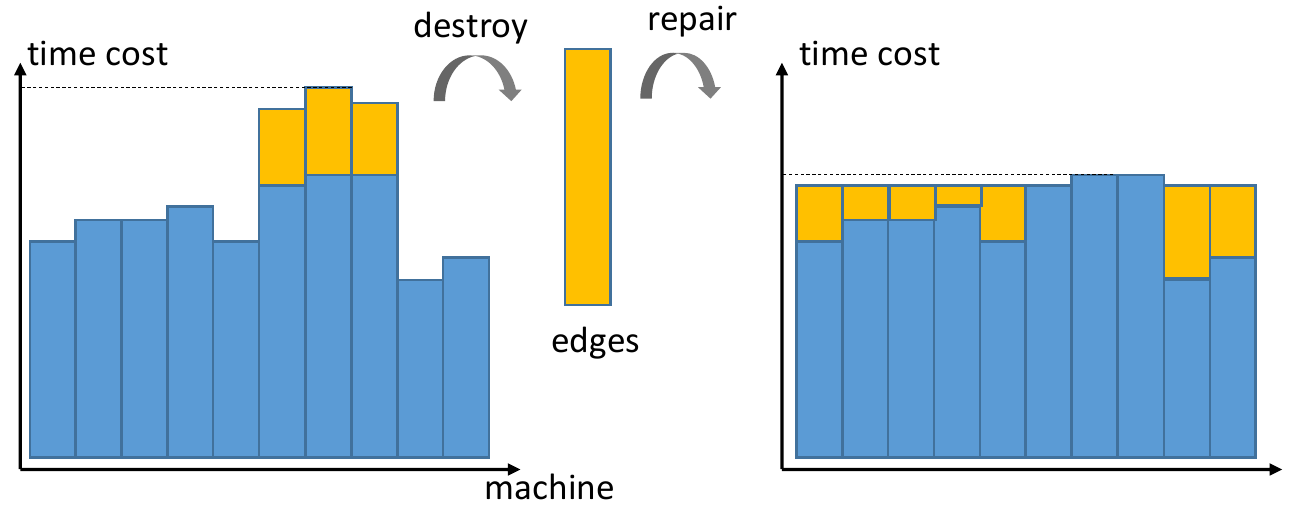}
	\vspace{-0.15in}
	\caption{destroy-repair}
	\label{fig:destroy-repair}
\end{figure}
\vspace{-0.15in}

In the repair step, the key point is to decide the most suitable machine for each destroyed edge when reconstructing the partition. 
A greedy heuristic is used to balance the assignment. 
For each destroyed edge $e$, we first select the machines where endpoints of $e$ already exist. 
If more than one machine is selected, we select the machine that has the lowest time cost.

 %We use parameter $\theta$ to control the proportion of destroyed edges, and parameter $\alpha$, $\beta$, $\gamma$ to control machine selection in destroy and repair process. $I=\{i|1\leq i\leq N\}$ is the index set of all machines.

\Paragraph{Re-partition}.
When the destroy-and-repair operator fails to improve the current partition, it may be a local optimum. 
Methodologically, diversification of SLS operators is needed to deal with this case, because different operators construct different types of neighborhood. 
In our method, a re-partition operator is applied.  %We apply the re-partition operator to change the partition in a larger scale to escape from local optima and to find a more promising partition. 
We specify $k$  subgraphs from  current partition $\{E_i|i \in \mathcal{I}, \mathcal{I}\subseteq\{ 1,2,...p\}, |\mathcal{I}| = k\}$, and apply Algorithm \ref{alg:expand} on the union of these subgraphs to get new  partitions while the unselected partitions keeps unchanged. 
In such a way the partitions among the selected machines is fully changed.
Hence, this move is in a scale larger than destroy-and-repair, and it is promising to get escaped from local optimum. 
When choosing the subset $\mathcal{I}$ for re-partition, we first select the  subgraph  $G_i$ with the largest $T_i=T_i^{cal}+T_i^{com}$, then select $k-1$ subgraphs $\{G_j\}$ with the maximum $n_{i,j}$ from remaining partitions, where  $n_{i,j}$ is the number of replica nodes between $G_i$ and $G_j$. 

\begin{algorithm}
 	\small
 	\caption{destroy-and-repair}
 	\label{alg:destroy}
 	\textbf{procedure} DestroyRepair($\{E_i|1\leq i\leq p\}$)  \\
 	% s_i 通信时间; t_i 计算时间; I 机器集合
 	
 	% \STATE \qquad $i \gets argmax_{1\leq j\leq N}[\textit{communication}(E_j)]$
 	% \STATE \qquad $q \gets k$-th largest value among $\{n_{i,j}|j \neq i\}$ 
 	% \STATE \qquad $I \gets \{j| n_{i,j} \geq q, j \neq i\}\cup\{i\}$
 	% \STATE \qquad $E' \gets EXPAND( \cup_{m\in I}E_m  , \delta)$
 	% \STATE \qquad \textbf{return}  $ E' \cup \{E_j| j \notin I\} $
 	
 	%$thd_{1} \gets min_{1\leq i\leq N}T_i^{com} + \alpha * (max_{1\leq i\leq N}T_i^{com} - min_{1\leq i\leq N}T_i^{com})$  \\
 	%$thd_{2} \gets min_{1\leq i\leq N}T_i^{cal} + \beta * (max_{1\leq i\leq N}T_i^{cal} - min_{1\leq i\leq N}T_i^{cal})$  \\
 	
 	\For{$i=1,2,...p$} {
 		$T_i=T_i^{com} + T_i^{cal}$
 	}
 	$thd \gets min_{1\leq i\leq p}T_i + \gamma * (max_{1\leq i\leq p}T_i - min_{1\leq i\leq p}T_i)$  \\
 	
 	$S \gets \emptyset$ \\
 	\For{$i=1,2,...p$}
 	{
 		\If{$T_i \geq thd$}   
 		{
 			$C \gets \textit{Remove and get a propotion of $\theta $ edges in } E_i$   \\
 			$S \gets S \cup C$   \\
 		}
 	}
 	$\{T_i^{com}\}, \{T_i^{cal}\}, \{T_i\} \gets UpdateObjective(\{E_i|1\leq i\leq p\})$  \\
 	\For{$\overline{xy} \in S$}
 	{
 		$A(x) \gets \{i|x\in V_i, 1\leq i\leq p\}$  \\
 		$A(y) \gets \{i|y\in V_i, 1\leq i\leq p\}$    \\
 		\If{$A(x) \cap A(y) \neq \emptyset$}
 		{
 			$i \gets BalancedGreedyRepair(A(x) \cap A(y), \{T_i\})$  \\
 		}
 		\If{$(i = 0) \ or\ (A(x) \cup A(y) \neq \emptyset\ and\ A(x) \cap A(y) = \emptyset)$}
 		{
 			$i \gets BalancedGreedyRepair(A(x) \cup A(y), \{T_i\})$  \\
 		}
 		\If{$(i = 0)\ or\ A(x) \cup A(y) = \emptyset$}
 		{
            $i \gets BalancedGreedyRepair(\{1,2,...p\}, \{T_i\})$  \\
 		}
 		$E_i\gets E_i\cup \{\overline{xy}\}$  \\
 		$\{T_i^{com}\}, \{T_i^{cal}\}, \{T_i\} \gets UpdateObjective(\{E_i|1\leq i\leq p\})$  \\
 	}
 	\If{objective $TC$ is improved}
 	{
 		\Return true   \\
 	}
 	\Else
 	{
 		\Return false \\
 	}
 	% \STATE \qquad $E_i\gets E_i\cup \{e_{x,y}\}$
 	% \STATE \qquad $\{s_i\}, \{t_i\} \gets UpdateObjs(\{E_i|1\leq i\leq N\})$
 \end{algorithm}
 
 \begin{algorithm}
 	\small
 	\caption{BalancedGreedyRepair}
 	\label{alg:greedy}
 	% \STATE \qquad $i \gets argmax_{1\leq j\leq N}[\textit{communication}(E_j)]$
 	% \STATE \qquad $q \gets k$-th largest value among $\{n_{i,j}|j \neq i\}$ 
 	% \STATE \qquad $I \gets \{j| n_{i,j} \geq q, j \neq i\}\cup\{i\}$
 	% \STATE \qquad $E' \gets EXPAND( \cup_{m\in I}E_m  , \delta)$
 	% \STATE \qquad \textbf{return}  $ E' \cup \{E_j| j \notin I\} $
 	
 	\textbf{procedure} BalancedGreedyRepair($S, \{T_i\}$)  \\
 	% s_i 通信时间; t_i 计算时间; I 机器集合
 	$v \gets infinity$  \\
 	$j \gets 0$	  \\
 	\For{$i \in S$}
 	{
 		
 		$criteria \gets T_i$ \\
 		
 		\If{$criteria < v$ and (memory of partition $i$ is enough)}
 		{
 			$v \gets criteria$  \\
 			$j \gets i$  \\
 		}
 	}
 	\Return $j$;
 \end{algorithm}

\Paragraph{Analysis}. 
In destroy-and-repair procedure (Algorithm \ref{alg:destroy}), $\theta |E_i|$ edges are selected to be reassigned.
For each edge, we need to greedily search the most suitable machine among all $p$ machines. 
Thus, the time complexity of destroy-and-repair is $O(p|E|)$ and the space complexity is $O(|E|)$.
In re-partition procedure (Algorithm \ref{alg:repart}), the worst-case time complexity is equivalent to Algorithm \ref{alg:expand} (i.e., $O(|E|+|V|log|V|)$), while the space complexity is $O(|V|+|E|)$.
In Algorithm \ref{alg:SLS}, the running time of the entire SLS algorithm is bounded by a threshold $T_0$, which limits the number of iterations that SLS performs.
% Larger T means longer local search, larger optimization space
In practice, the number of iterations is usually a small constant (<10), thus the time complexity is $O(p|E|+|E|+|V|\log|V|)=O(p|E|+|V|\log|V|)$.

 \begin{algorithm}
 	\small
 	\caption{Re-partition}
 	\label{alg:repart}
 		\textbf{procedure} REPARTITION($\{E_i|1\leq i\leq p\}, k$)  \\
        $i \gets argmax_{1\leq j\leq p}\{T_j\}$  \\
 		$q \gets k$-th largest value among $\{n_{i,j}|j \neq i\}$    \\
 		$\mathcal{I} \gets \{j| n_{i,j} > q, j \neq i\}\cup\{i\}$    \\
 		$E' \gets EXPAND( \cup_{m\in \mathcal{I}}E_m  , \delta)$    \\
 		\Return  $ E' \cup \{E_j| j \notin \mathcal{I}\} $  \\
 \end{algorithm}

To sum up, the total time complexity of WindGP can be computed by accumulating three phases:
\begin{equation}
    \small
	\begin{aligned}
       Time & = O(p^2+\sum_i(|E_i|+|V_i|\log{|V_i|)}+T) \\
            %& = O(p^2+\sum_i(|E_i|)+\sum_i(|V_i|\log{|V|)}+p\theta |E|+|V|\log|V|) \\
            & =O(p^2+p|E|+|V|\log|V|) \\
	\end{aligned}
\end{equation}
And the space complexity can be formulized as:
\begin{equation}
    \small
	\begin{aligned}
          Space & = O(p+max_i{|E_i|+|V|}+|E|+|V|+|E|) \\
                & = O(p+|V|+|E|) \\
	\end{aligned}
\end{equation}
Obviously, both the time complexity and the space complexity of WindGP are linear to the graph size.

\section{Extensions}\label{sec:extension}

In this section, we discuss several extensions of our algorithm (WindGP).
Though WindGP focuses on solving the problem defined in Definition \ref{def:problem}, it can be generalized to more scenarios such as \emph{directed labelled graph}, \emph{vertex-centric partition} and \emph{Map-Reduce based system}.

\Paragraph{Directed Graph and Labelled Graph}.
For directed graphs, WindGP can be adjusted by distinguishing incoming/outgoing edges as well as in/out degrees in Algorithm \ref{alg:expand}, i.e., the way of graph traversal changes.
Real-life graphs may contain many node/edge labels (e.g., properties or features), which also need to be stored and computed.
In order to process these property graphs, we maintain the mapping ($f$) between node/edge ID and label sets.
The graph structure is partitioned first, then the label sets are assigned to corresponding partition by the mapping $f$.
Note that the calculation costs of a node and an edge (i.e., $C_i^{node}$ and $C_i^{edge}$, respectively) as well as the memory occupation ($M^{node}$ and $M^{edge}$) need to be increased according to the size of label sets.
% NOTE: if the memory occupation of nodes(or edges) are different, the formula will be more complex and preprocessing needs to be changed.

\Paragraph{Vertex-Centric Partition}.
Except for edge-centric partition, vertex-centric partition is also frequently used in real applications.
However, the research of edge-cut on heterogeneous machines is still blank.
Though WindGP is originally designed for vertex-cut, it can adapt to edge-cut simply.
Firstly, WindGP needs to generate the result of edge partition (i.e., $p$ partitions).
For each vertex $u$, it may exist in several machines.
Let $deg_i(u)$ and $deg(u)$ be the degree of $u$ in $G_i$ and $G$, respectively,
We need to re-assign all vertices to $p$ machines according to the current edge partition.
Intuitively, if we place vertex $u$ in machine $j$ with larger partial degree of $u$, the \emph{edge-cut} of $u$ should be smaller.
Thus, each vertex $u$ should be placed in the $k$-th machine with the maximum $\frac{deg_k(u)}{deg(u)+1}$ as long as machine $k$ is not full.
Finally, we enumerate all edges of $G$ and put each edge $\overline{uv}$ in the partitions that $u$ and $v$ belong to.

\Paragraph{Map-Reduce based system}.
The routine of distributed running is different on map-reduce based systems such as GraphX \cite{GraphX} and Giraph \cite{Giraph}.
The communications can only occur after all local calculations are over, as shown in Figure \ref{fig:BSP2}.
In this case, the execution time should be $max_{i}(max_{i}(T_{cal}^{i})+T_{com}^{i})$, which is different from $TC$ defined in Definition \ref{def:problem}.
This can also be processed well by our 3-phase mechanism, while the only difference is the object goal in the post-processing phase.

\begin{figure}[htbp]
	\centering
	\includegraphics[width=7cm]  {\picfolder 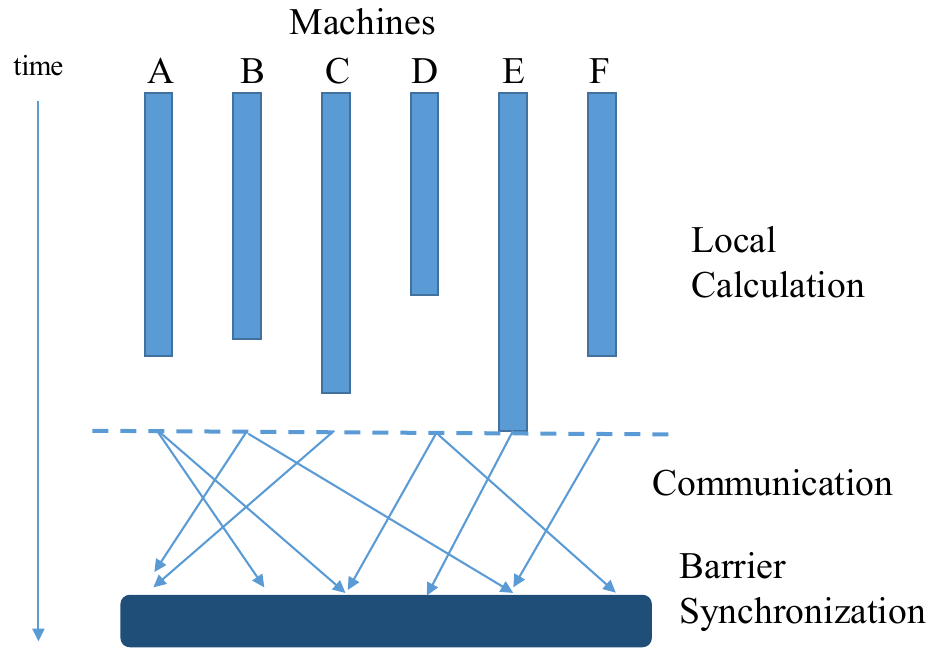}        	
    \vspace{-0.1in}     	
	\caption{The routine of Map-Reduce based Graph Engine}      	
	\label{fig:BSP2}  
\end{figure}

Streaming graph partition on heterogeneous machines is also fascinating, but it is rather different from our algorithms, thus we leave it to future work.
Besides, modern graphs have tens of billions of edges or trillions of attributes, which is hard to be partitioned by a single machine.
For example, Graph Neural Network (GNN) usually has embeddings with hundreds of features \cite{AGL}.
Therefore, distributed graph partition on heterogeneous machines is also an interesting direction. 

Except for the variants, we can also extend WindGP algorithm to support multi-threading for speedup.
For example, Algorithm \ref{alg:alloc} , Algorithm \ref{alg:destroy} and Algorithm \ref{alg:greedy} can be accelerated by \emph{OpenMP} to utilize multiple cores in the running machine.

\section{Experiments}\label{sec:experiment}

We evaluate our method WindGP as well as all counterparts on a machine running CentOS 7 system with 24 Intel E5-2690 2.60GHz cores, 128GB memory and 1TB disk.
WindGP is implemented with \emph{C++} and compiled with \emph{g++ 7.3.0}, and the optimization flag is set to \emph{-O2}. 
According to the experiments of previous work \cite{NE, EBV}, we select four state-of-the-art algorithms as counterparts: METIS \cite{METIS}, HDRF \cite{HDRF}, NE \cite{NE}, EBV \cite{EBV}.
To provide a fair comparison, we modify these algorithms to meet the requirement of heterogeneous-machine edge partition, i.e., adding constraints of memory capacity of each machine.
Note that METIS is originally a vertex-centric method, thus we transform it into the edge-centric solution in the same way discussed in \cite{NE}: with the node degree as the node weight, it partitions $G$ using $gpmetis$ command, then assigns each edge $\overline{uv}$ to the machine that $u$ or $v$ exists randomly as long as the machine has enough memory.
In our experiments, the main metric is the $TC$ score (Definition~\ref{def:problem}), which measures the quality of partition results.
Besides, the running time of distributed graph algorithms on different partition results is also evaluated in Section \ref{sec:distributed}.
The comparison with heterogeneous solutions \cite{HeterCompPart,GrapH,HaSGP,HAEP} is included in the final part.

\begin{table}[htbp]
    \small
    \caption{Statistics of Datasets}
    \label{tab:dataset}
    \vspace{-0.1in}
    \begin{threeparttable}
        %\small
        \centering
    \setlength{\tabcolsep}{3mm}{
        \begin{tabular}{crrrr}
            \toprule
            Name & $|V|$ & $|E|$ & MD\tnote{1} & Type\tnote{2} \\
            \midrule
	CO & 3,072,441 & 117,185,083 & 33,313 & rs \\ 
	LJ & 4,847,570 & 33,099,465  & 20,290 & rs \\ 
    PO & 1,632,803 & 30,622,564 & 20,518  & rs \\
	CP & 3,774,768 & 16,518,947 & 793 & rs \\ 
    RN  & 1,965,206 & 2,766,607 & 8 & rm \\
    \midrule
    TW & 41,652,230 & 1,202,513,046 & 3M & rs \\ 
    DB & 233M & 1.1B  & 17M & rs  \\
    FR & 65M & 1.8B & 5.2K & rs \\
    YH & 417M & 2.8B & 2.5K & rs \\
    %WD & 97M & 1B & 6.7M & s \\
            \bottomrule
        \end{tabular}
    }
        \begin{tablenotes}
            \item[1] Maximum degree of the graph.
            \item[2] Graph type: r:real-world, s:scale-free, and m:mesh-like.
        \end{tablenotes}
    \end{threeparttable}
\end{table}

\subsection{Dataset}\label{sec:dataset}

We mainly evaluate all solutions on some representative dataset in the Stanford Network Analysis Project (SNAP \cite{site:snap}) such as TW (Twitter \cite{twitter2010}), CO (com-Orkut),  LJ(soc-LiveJournal), PO (soc-Pokec), CP (cit-Patents \cite{patent}) and RN (roadNet-CA). 
The details of these datasets are listed in Table \ref{tab:dataset}.
Besides, we also use a synthetic generator (R-MAT \cite{R-MAT}) to generate a series of power-law graphs for the scalability test in Section \ref{sec:scalability}.
In Section \ref{sec:distributed}, graphs with billions of edges are used for distributed graph computing because modern graph and machine resource are usually very large.

To eliminate the influence of randomness in some algorithms (e.g., NE \cite{NE}), all tasks are run 10 times and the averaged result is recorded as the metric.
In our experiments, the executing time of all partition algorithms is required to be smaller than 1 hour (for graphs with billions of edges) and 10 minutes (for other datasets), respectively.

\Paragraph{Machine Configuration}. 
\nop{
Refer to cloud platform \cite{DBLP:conf/ic2e/GarraghanTX13,DBLP:journals/tompecs/HerbstBKOEKEKBA18,DBLP:journals/simpra/ZakaryaG19}, the resources of machines can be quantified by relative rates.
For example, machine $i$ has the smallest memory (100GB), then we can set $M_i=1$ and set $M_j=2$ if machine $j$ has 200GB memory.
It is similar for the computing ($cores\times frequency$) and communication (network bandwidth).
}
The quantification of machine resource is detailed in Section \ref{sec:problem}, and experiments are done on different homogeneous clusters to verify the feasibility of our methodology.
%More thorough study should include cache size.

We mainly use two types of machines: \emph{super machine} and \emph{normal machine}.
For simplicity, $M^{node}$ and $M^{edge}$ are set to 1 and 2 respectively.
Recall that in Definition \ref{def:problem}, the configuration of each machine is a quadruple: $(M_i, C_i^{node},C_i^{edge},C_i^{com})$.
On large graphs (e.g., TW and CO), there are 100 machines (20 super machines and 80 normal machines), and the configuration of super machine and normal machine is $(10^8, 10, 15, 15)$ and $(3\times 10^7, 5, 10, 10)$, respectively.
In contrast, there are 30 machines (10 super machines and 20 normal machines) on other datasets and the configuration is $(10^7, 10, 15, 15)$ and $(3\times 10^6, 5, 10, 10)$.
In Section \ref{sec:scalability}, we also study the impact of the number of machine types and the machine number.

\Paragraph{Parameter Tuning}. 
Several hyper-parameters are used in our solution as well as others.
For each algorithm, comprehensive tests are conducted on all datasets and the best parameter is found out for each dataset.
In WindGP, $\alpha$ and $\beta$ are both set to 0.3, and other parameters like $\gamma$ and $\theta$ keep the default value as specified in Section \ref{sec:algorithm}.
The details of hyper-parameter tuning in WindGP is listed in Table \ref{tab:tune-alpha}, \ref{tab:tune-beta}, \ref{tab:tune-gamma}, \ref{tab:tune-theta}, \ref{tab:tune-N0},  and \ref{tab:tune-T0}.
Lookup in Table \ref{tab:notations} for the meaning of these hyper-parameters.

As for $\alpha$, it controls the balance between $N(u)\setminus S$ and $N(u)\cap S$.
On LJ, CP and RN, the average degree is low (<7), thus the long-tail effect caused by small $N(u)\cap S$ is more prominent.
While on other high-degree graphs, much larger $N(u)\setminus S$ generates too many borders which raises the communication cost. 
Overall, the best value is set to 0.3 for $\alpha$.

\begin{table*}[htbp]
	\small
	\caption{Tuning of $\alpha$ in WindGP}
	\label{tab:tune-alpha}
	\vspace{-0.1in}
	\begin{threeparttable}
		\small
		\centering
			\begin{tabular}{ |c|c|c|c|c|c|c|c|c|c|c|} 
				\hline
				$TC$ & 0 &  0.1 & 0.2 & 0.3 & 0.4 & 0.5 & 0.6 & 0.7 & 0.8 & 0.9 \\
				\hline
				TW  & 64M & 62M & 61M & 60M & 61M & 62M & 65M & 70M & 76M & 85M \\ 
				\hline
				CO  & 34M & 32M & 31M & 31M & 32M & 33M & 36M & 39M & 45M & 52M \\ 
				\hline	
				LJ  & 34M & 28M & 25M & 23M & 23M & 24M & 25M & 25M & 27M & 28M \\ 
				\hline
				PO  & 25M & 23M & 21M & 21M & 21M & 23M & 25M & 29M & 33M & 38M \\ 
				\hline
				CP  & 20M & 15M & 12M & 11M & 11M & 11M & 12M & 13M & 13M & 14M \\ 
				\hline
				RN  & 26M & 20M & 17M & 15M & 15M & 16M & 18M & 18M & 19M & 19M \\ 
				\hline			
			\end{tabular}
			%}
	\end{threeparttable}
\end{table*}

As for $\beta$, it controls the impact of border vertices.
On graphs with low average degree, border vertices are much fewer, thus the impact of border vertex number and $N(u)\setminus S$ are both slight.
In contrast, there are too many borders on TW, CO and PO, raising $TC$ in two directions. 
Overall, the best value is set to 0.3 for $\beta$.

\begin{table*}[htbp]
	\small
	\caption{Tuning of $\beta$ in WindGP}
	\label{tab:tune-beta}
	\vspace{-0.1in}
	\begin{threeparttable}
		\small
		\centering
			\begin{tabular}{ |c|c|c|c|c|c|c|c|c|c|c|} 
				\hline
				$TC$ & 0 &  0.1 & 0.2 & 0.3 & 0.4 & 0.5 & 0.6 & 0.7 & 0.8 & 0.9 \\
				\hline
				TW  & 80M & 71M & 65M & 60M & 61M & 62M & 64M & 64M & 67M & 70M \\ 
				\hline
				CO  & 42M & 37M & 33M & 31M & 32M & 33M & 33M & 33M & 34M & 34M \\ 
				\hline	
				LJ  & 24M & 24M & 23M & 23M & 24M & 25M & 25M & 25M & 26M & 27M \\ 
				\hline
				PO  & 30M & 25M & 22M & 21M & 22M & 22M & 24M & 25M & 25M & 26M \\ 
				\hline
				CP  & 13M & 13M & 12M & 11M & 11M & 12M & 12M & 13M & 14M & 15M \\ 
				\hline
				RN  & 15M & 15M & 15M & 15M & 15M & 15M & 16M & 17M & 17M & 17M \\ 
				\hline			
			\end{tabular}
			%}
	\end{threeparttable}
\end{table*}

According to Algorithm \ref{alg:destroy} in Section \ref{sec:post-process}, $\gamma$ is the threshold of edges above which partitions should be destroyed.
When $\gamma$ is set to 1, only partitions with the maximum cost will be destroyed, which limits the improvement of subgraph-local search.
Once $\gamma$ decreases, $TC$ drops rather slightly with much higher executing time.
For example, the best $TC$ is achieved by $\gamma=0$, but the speedup is <9\%, while the executing time of post-processing is >10 times longer.
Besides, as the subgraph-local search needs several iterations, only a small percentage of partitions should be destroyed in each iteration. 
Thus, the best value is set to 0.9 for $\gamma$.

\begin{table*}[htbp]
	\small
	\caption{Tuning of $\gamma$ in WindGP}
	\label{tab:tune-gamma}
	\vspace{-0.1in}
	\begin{threeparttable}
		\small
		\centering
			\begin{tabular}{ |c|c|c|c|c|c|c|c|c|c|c|c|} 
				\hline
				$TC$ &  0 & 0.1 & 0.2 & 0.3 & 0.4 & 0.5 & 0.6 & 0.7 & 0.8 & 0.9 & 1 \\
				\hline
				TW  & 58M & 58M & 58M & 58M & 59M & 59M & 60M & 60M & 60M & 60M & 68M \\ 
				\hline
				CO  & 30M & 30M & 30M & 30M & 31M & 31M & 31M & 31M & 31M & 31M & 33M \\ 
				\hline	
				LJ  & 22M & 22M & 22M & 22M & 22M & 23M & 23M & 23M & 23M & 23M & 25M \\ 
				\hline
				PO  & 20M & 20M & 20M & 20M & 21M & 21M & 21M & 21M & 21M & 21M & 23M \\ 
				\hline
				CP  & 10M & 10M & 11M & 11M & 11M & 11M & 11M & 11M & 11M & 11M & 12M \\ 
				\hline
				RN  & 15M & 15M & 15M & 15M & 15M & 15M & 15M & 15M & 15M & 15M & 16M \\ 
				\hline			
			\end{tabular}
			%}
	\end{threeparttable}
\end{table*}

To select a appropriate proportion of edges in each destroyed machine, we vary $\gamma$ from 0.002 to 0.02 and evaluate $TC$.
Obviously, too small $\gamma$ limits the optimization space of \emph{destroy-and-repair}, while $\gamma$ larger than 0.1 does not bring extra gain.
In conclusion, $\theta$ is set to 1\% by default.

\begin{table*}[htbp]
	\small
	\caption{Tuning of $\theta$ in WindGP}
	\label{tab:tune-theta}
	\vspace{-0.1in}
	\begin{threeparttable}
		\small
		\centering
			\begin{tabular}{ |c|c|c|c|c|c|c|c|c|c|c|} 
				\hline
				$TC$ &  0.002 & 0.004 & 0.006 & 0.008 & 0.01 & 0.012 & 0.014 & 0.016 & 0.018 & 0.02 \\
				\hline
				TW  & 67M & 65M & 64M & 63M & 60M & 60M & 60M & 60M & 60M & 60M \\ 
				\hline
				CO  & 40M & 38M & 35M & 33M & 31M & 31M & 31M & 31M & 31M & 31M \\ 
				\hline	
				LJ  & 25M & 25M & 24M & 24M & 23M & 23M & 23M & 23M & 23M & 23M \\ 
				\hline
				PO  & 26M & 24M & 23M & 23M & 21M & 21M & 21M & 21M & 21M & 21M \\ 
				\hline
				CP  & 12M & 12M & 12M & 11M & 11M & 11M & 11M & 11M & 11M & 11M \\ 
				\hline
				RN  & 16M & 15M & 15M & 15M & 15M & 15M & 15M & 15M & 15M & 15M \\ 
				\hline			
			\end{tabular}
			%}
	\end{threeparttable}
\end{table*}

As for the number of local and global try times $N_0$ and $T_0$, a thorough study is conducted that varies them from 1 to 9 for each dataset.
To sum up, it is enough to set $N_0=5$, and the best value of $T_0$ varies on different graphs.

\begin{table}[htbp]
	\small
	\caption{Tuning of $N_0$ in WindGP}
	\label{tab:tune-N0}
	\vspace{-0.1in}
	\begin{threeparttable}
		\small
		\centering
			\begin{tabular}{ |c|c|c|c|c|c|c|c|c|c|} 
				\hline
				$TC$ &  1 & 2 & 3 & 4 & 5 & 6 & 7 & 8 & 9 \\
				\hline
				TW  & 68M & 64M & 62M & 61M & 60M & 60M & 60M & 60M & 60M \\ 
				\hline
				CO  & 35M & 34M & 32M & 32M & 31M & 31M & 31M & 31M & 31M \\ 
				\hline	
				LJ  & 25M & 24M & 23M & 23M & 23M & 23M & 23M & 23M & 23M \\ 
				\hline
				PO  & 25M & 24M & 22M & 22M & 21M & 21M & 21M & 21M & 21M \\ 
				\hline
				CP  & 12M & 12M & 11M & 11M & 11M & 11M & 11M & 11M & 11M \\ 
				\hline
				RN  & 16M & 15M & 15M & 15M & 15M & 15M & 15M & 15M & 15M \\ 
				\hline			
			\end{tabular}
			%}
	\end{threeparttable}
\end{table}

\begin{table}[htbp]
	\small
	\caption{Tuning of $T_0$ in WindGP}
	\label{tab:tune-T0}
	\vspace{-0.1in}
	\begin{threeparttable}
		\small
		\centering
			\begin{tabular}{ |c|c|c|c|c|c|c|c|c|c|} 
				\hline
				$TC$ &  1 & 2 & 3 & 4 & 5 & 6 & 7 & 8 & 9 \\
				\hline
				TW  & 70M & 66M & 64M & 63M & 61M & 61M & 60M & 60M & 60M \\ 
				\hline
				CO  & 35M & 33M & 32M & 32M & 32M & 31M & 31M & 31M & 31M \\ 
				\hline	
				LJ  & 26M & 24M & 24M & 23M & 23M & 23M & 23M & 23M & 23M \\ 
				\hline
				PO  & 25M & 23M & 22M & 22M & 21M & 21M & 21M & 21M & 21M \\ 
				\hline
				CP  & 12M & 11M & 11M & 11M & 11M & 11M & 11M & 11M & 11M \\ 
				\hline
				RN  & 17M & 15M & 15M & 15M & 15M & 15M & 15M & 15M & 15M \\ 
				\hline		
			\end{tabular}
			%}
	\end{threeparttable}
\end{table}

\subsection{Evaluation by the metric $TC$}\label{sec:evaluate}

\Paragraph{Analysis of optimization techniques}.
Recall that WindGP propose three novel techniques: preprocessing capacity, best-first search, post-processing.
To evaluate the efficiency of each technique in WindGP, we conduct the experiments on six real graphs and depict the results in Figure \ref{fig:tc-windgp}.
Let WindGP be the entire solution, WindGP$^+$ denotes the solution that removes post-processing from WindGP.
Similarly, WindGP$^*$ denotes the solution that removes best-first search from WindGP$^+$, while WindGP$^-$ is the naive solution without all optimization techniques, which iteratively explores the graph to form a partition as long as the machine memory and load balance ratio permit.
Its expansion scheme is similar to NE \cite{NE}, which tends to form a subgraph with high cohesion by connectivity-based expansion.
In order to view the details of all partitions precisely, the histograms of the costs of different partitions on CP and LJ are also listed in Figure \ref{fig:hist-CP}, Figure \ref{fig:hist-LJ} and Figure \ref{fig:hist-CO}, respectively.

\begin{figure}[htbp]
	\centering
	\includegraphics[width=8cm]  {\picfolder 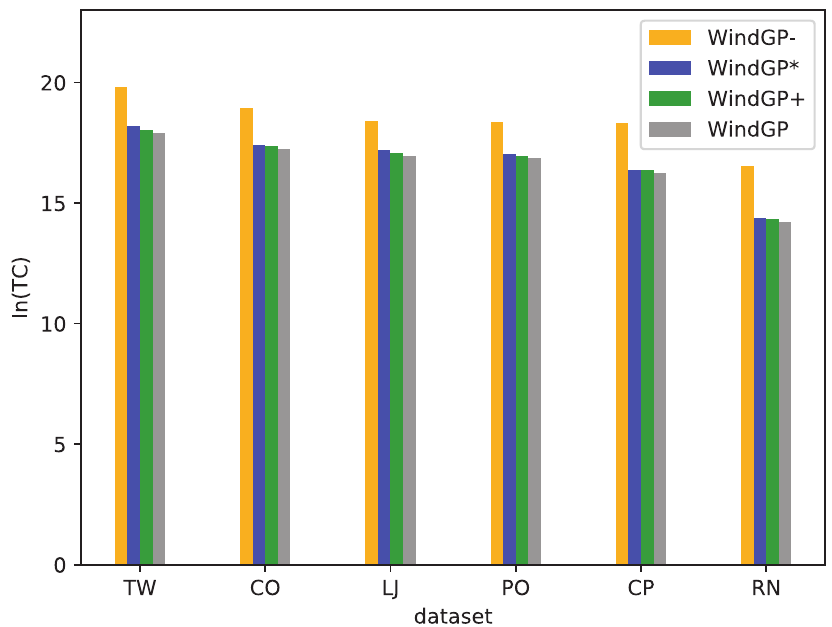}        	
    \vspace{-0.1in}     	
	\caption{Comparison of different techniques in WindGP ($\ln TC$ is the logarithm of $TC$)}      	
	\label{fig:tc-windgp}  
\end{figure}

Obviously, $TC$ of WindGP$^*$ is much lower than WindGP$^-$, and the speedup is $5\times$ on TW, $4.5\times$ on CO, $3.3\times$ on LJ, $3.8\times$ on PO, $7.5\times$ on CP, and $8.8\times$ on RN, respectively.
This prominent improvement comes from the pre-computed capacity of edges in each machine.
Different from the naive solution that only considers the upper bound of machine memory  and homogeneous load balance ratio, our strategy combines the computation cost and the memory size of all machines to generate a nearly best plan that can balance the computation cost across these heterogeneous machines (as shown in Figure \ref{fig:hist2-CP} and Figure \ref{fig:hist2-LJ}).
Note that traditional partition methods use $\frac{|E|}{p}$ as the capacity of edges, which can not be used in heterogeneous scenarios because the memory of some machines may be not enough while other machines may have memory size larger than $\frac{|E|}{p}$.
Besides, the traditional threshold does not distinguish the difference of node/edge computation cost between heterogeneous machines.
Recall that in Section \ref{sec:preprocess} we deduce the error bound $\frac{p^2}{|E|}$, and the experiment on small graphs with hundreds of edges shows that the difference between our solution and the optimal is always within 5\%, which can be further refined by the post-processing.
According to the theoretical bound, on larger graphs the error should be much smaller.

The second technique (i.e., the best-first search) further boosts the performance by optimizing the expansion process, leading to $>1.1\times$ on TW, CO and LJ.
The performance gain is mainly acquired by lowering the total communication cost (as shown in Figure \ref{fig:hist3-CP} and Figure \ref{fig:hist3-LJ}).
The degree balanced generation as well as border generation can help reduce the number of cross-machine vertices and improve the subgraph cohesion.
However, on low-degree graphs (CP and RN), the speedup is rather limited (only $\sim1.04\times$) due to the small percentage of communication cost.
On CP, the percentage of communication cost is $<30\%$ on sparse graphs (i.e., graphs with low average degree $\frac{|E|}{|V|}$), while it is $>50\%$ on LJ and CO.
Especially, the effect is tiny on RN because it is mesh-like graph and its structure is naturally balanced.

Finally, the post-processing technique (subgraph-local search) brings $>1.15\times$ speedup by flattening the total cost between all machines.
Note that the previous two techniques balance the computation cost and the communication cost respectively, however, they do not balance the total cost well.
In fact, the effect of best-first search is sometimes restricted because it can not disturb the balance of computation cost.
Besides, during expansion the information is not complete (e.g., the first partition has no communication cost), which further drags down the performance.
Subgraph-local search breaks this restriction by allowing the computation cost to be imbalanced between heterogeneous machines, as long as the total cost drops.
Based on the global view with complete information, it moves or swaps edges between machines to lower the highest total cost.

%\nop{
\begin{figure*}[htbp]
    %\small 
    \centering  
    \subfigure[{\tiny WindGP$^-$}]
    {
        \label{fig:hist1-CP}
        \begin{minipage}[c]{4cm}        
            \centering      
            \includegraphics[width=4cm]{\picfolder 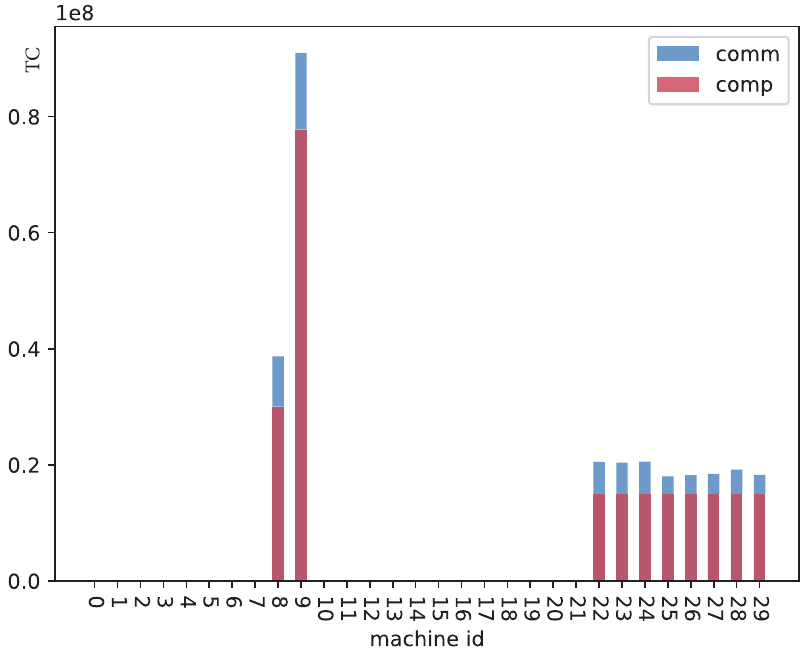}     
        \end{minipage}  
    }
    \subfigure[{\tiny WindGP$^*$}]
    {
        \label{fig:hist2-CP}
    \begin{minipage}[c]{4cm}        
        \centering      
        \includegraphics[width=4cm]{\picfolder 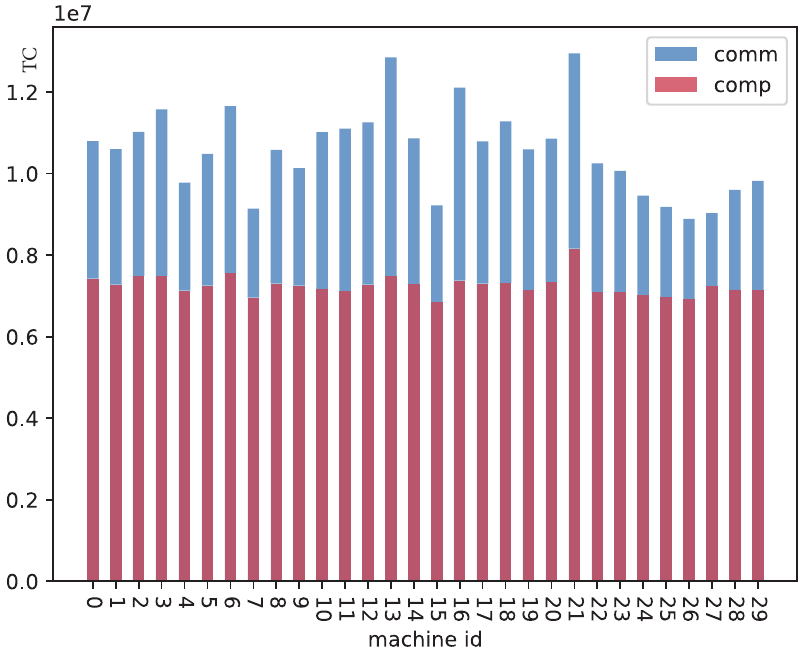}       
    \end{minipage}  
    }
    \subfigure[{\tiny WindGP$^+$}]
    {
        \label{fig:hist3-CP}
    \begin{minipage}[c]{4cm}        
        \centering      
        \includegraphics[width=4cm]{\picfolder 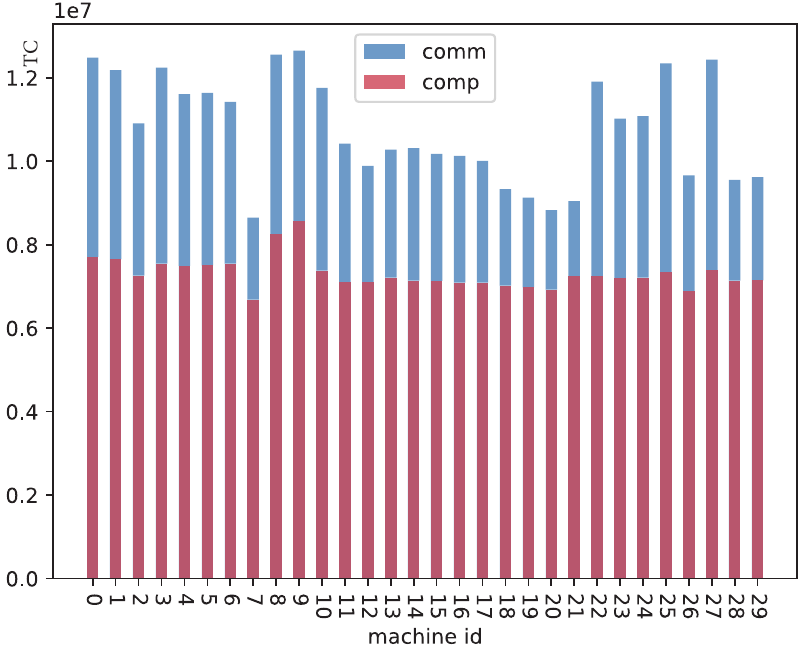}       
    \end{minipage}  
    }
    \subfigure[{\tiny WindGP}]
    {
        \label{fig:hist4-CP}
    \begin{minipage}[c]{4cm}        
        \centering      
        \includegraphics[width=4cm]{\picfolder 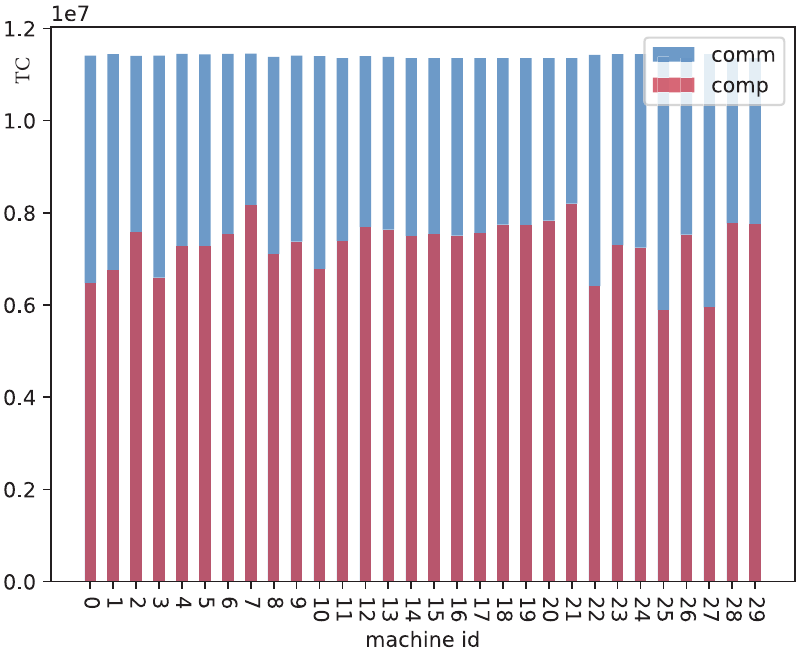}       
    \end{minipage}  
    }
    \vspace{-0.1in}
    \caption{The histogram of partitions of WindGP on CP}
	\label{fig:hist-CP}  
\end{figure*}
%}

\begin{figure*}[htbp]
    %\small 
    \centering  
    \subfigure[{\tiny WindGP$^-$}]
    {
        \label{fig:hist1-LJ}
        \begin{minipage}[c]{4cm}        
            \centering      
            \includegraphics[width=4cm]{\picfolder 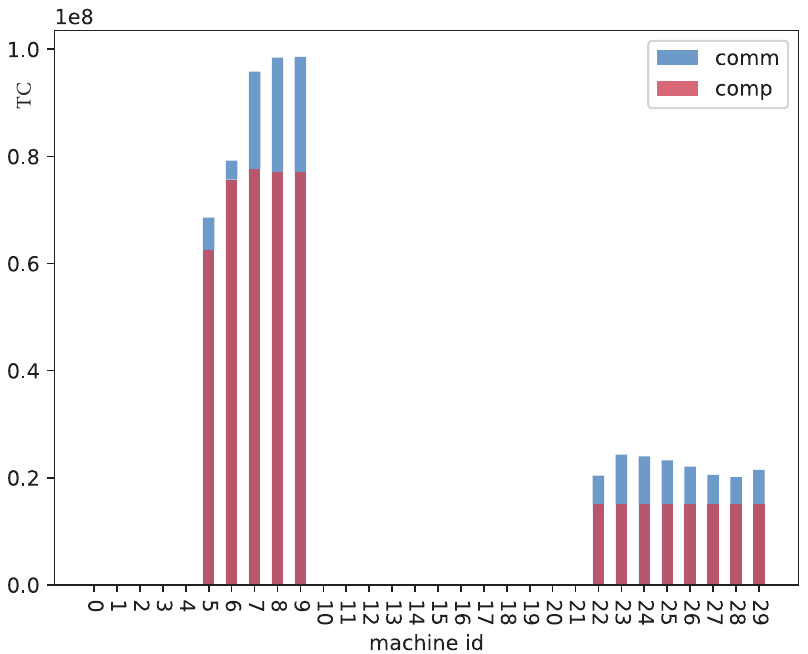}     
        \end{minipage}  
    }
    \subfigure[{\tiny WindGP$^*$}]
    {
        \label{fig:hist2-LJ}
    \begin{minipage}[c]{4cm}        
        \centering      
        \includegraphics[width=4cm]{\picfolder 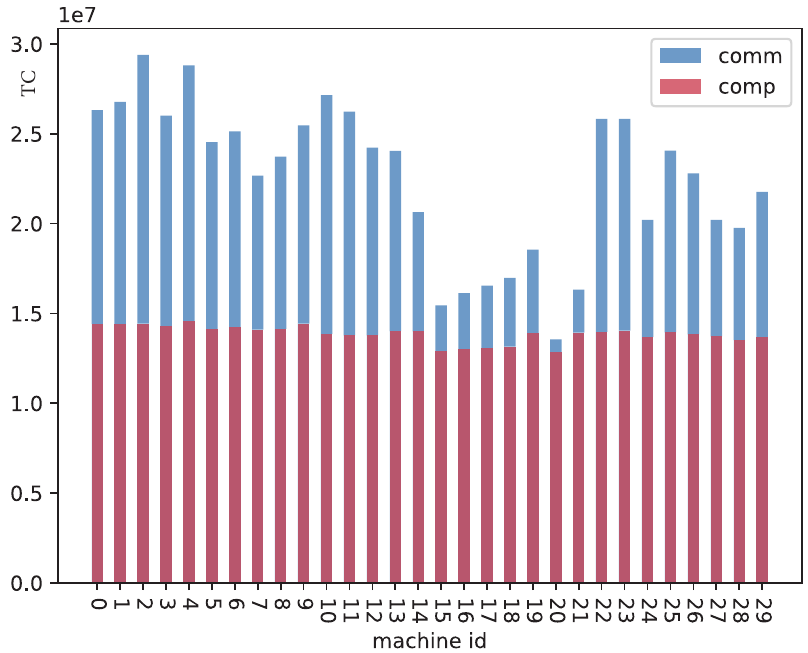}       
    \end{minipage}  
    }
    \subfigure[{\tiny WindGP$^+$}]
    {
        \label{fig:hist3-LJ}
    \begin{minipage}[c]{4cm}        
        \centering      
        \includegraphics[width=4cm]{\picfolder 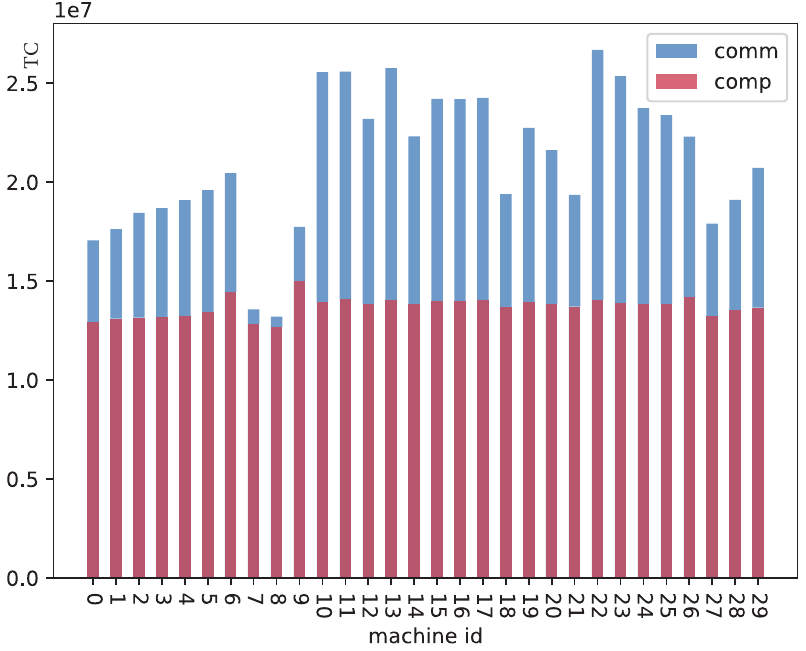}       
    \end{minipage}  
    }
    \subfigure[{\tiny WindGP}]
    {
        \label{fig:hist4-LJ}
    \begin{minipage}[c]{4cm}        
        \centering      
        \includegraphics[width=4cm]{\picfolder 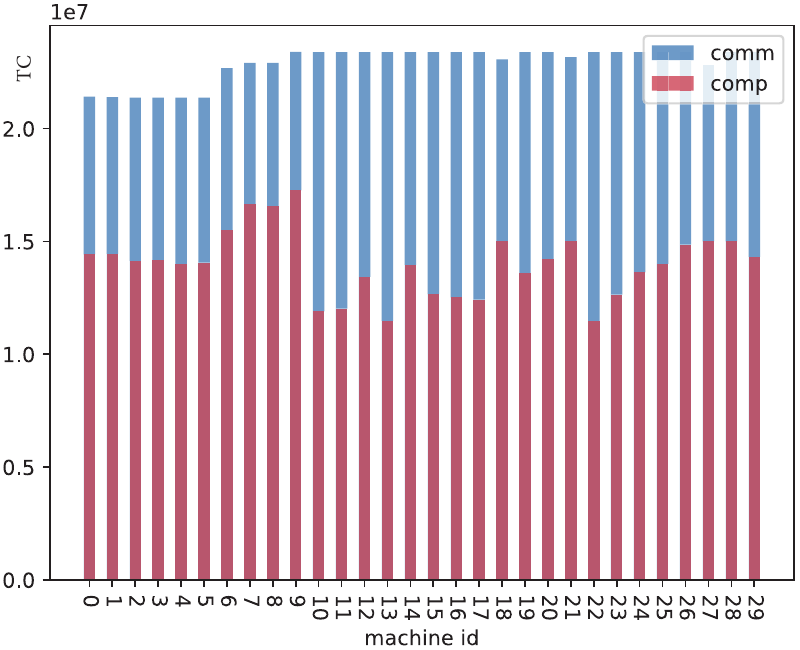}       
    \end{minipage}  
    }
    \vspace{-0.2in}
    \caption{The histogram of partitions of WindGP on LJ}
	\label{fig:hist-LJ}  
\end{figure*}

% NOTICE: the figure font is too small
%\nop{
\begin{figure*}[htbp]
    %\small 
    \centering  
    \subfigure[{\tiny WindGP$^-$}]
    {
        \label{fig:hist1-CO}
        \begin{minipage}[c]{4cm}        
            \centering      
            \includegraphics[width=4cm]{\picfolder 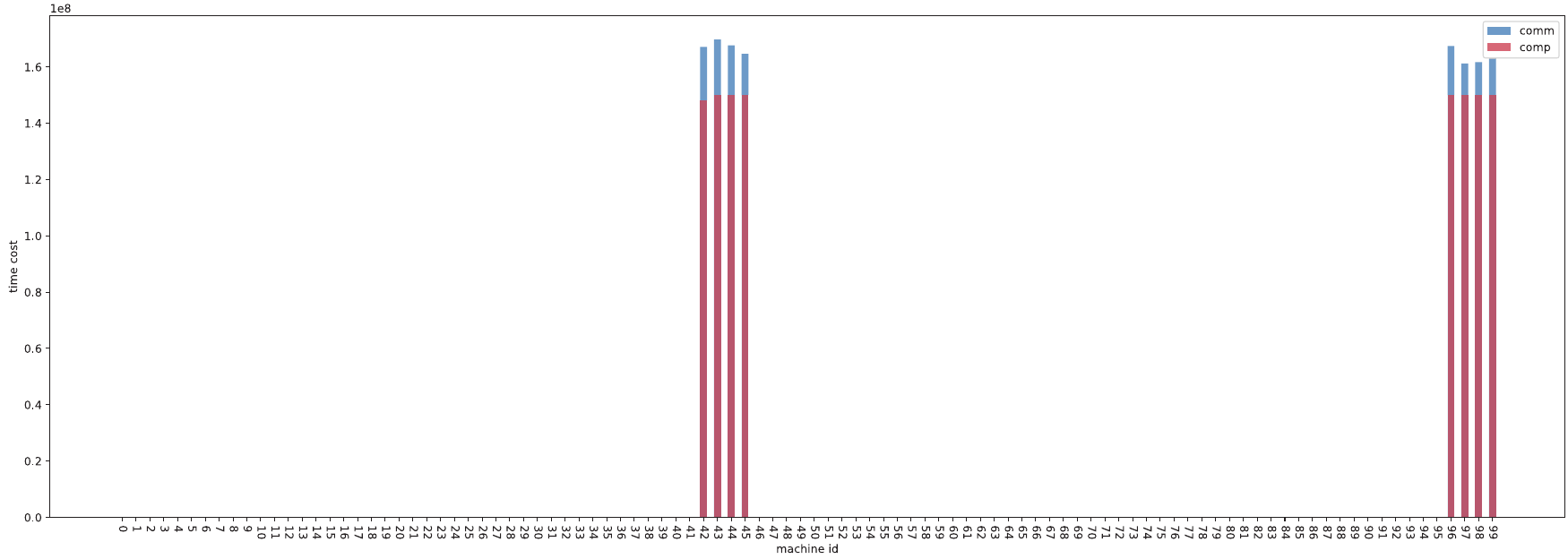}     
        \end{minipage}  
    }
    \subfigure[{\tiny WindGP$^*$}]
    {
        \label{fig:hist2-CO}
    \begin{minipage}[c]{4cm}        
        \centering      
        \includegraphics[width=4cm]{\picfolder 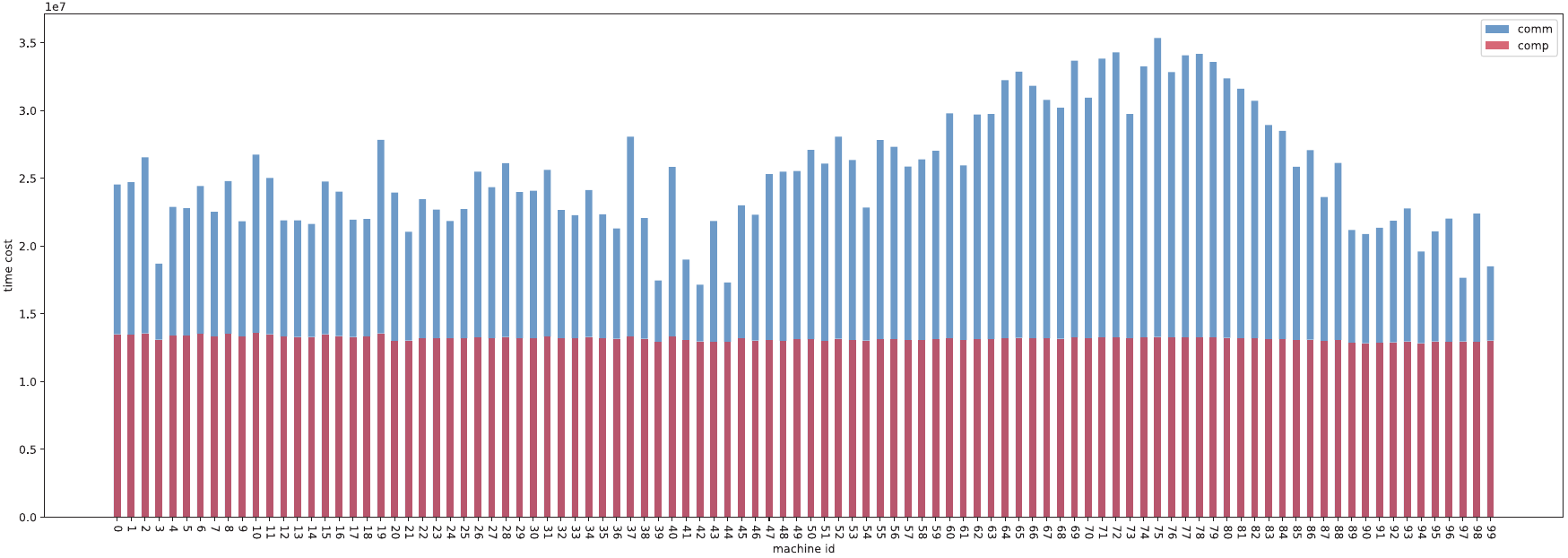}       
    \end{minipage}  
    }
    \subfigure[{\tiny WindGP$^+$}]
    {
        \label{fig:hist3-CO}
    \begin{minipage}[c]{4cm}        
        \centering      
        \includegraphics[width=4cm]{\picfolder 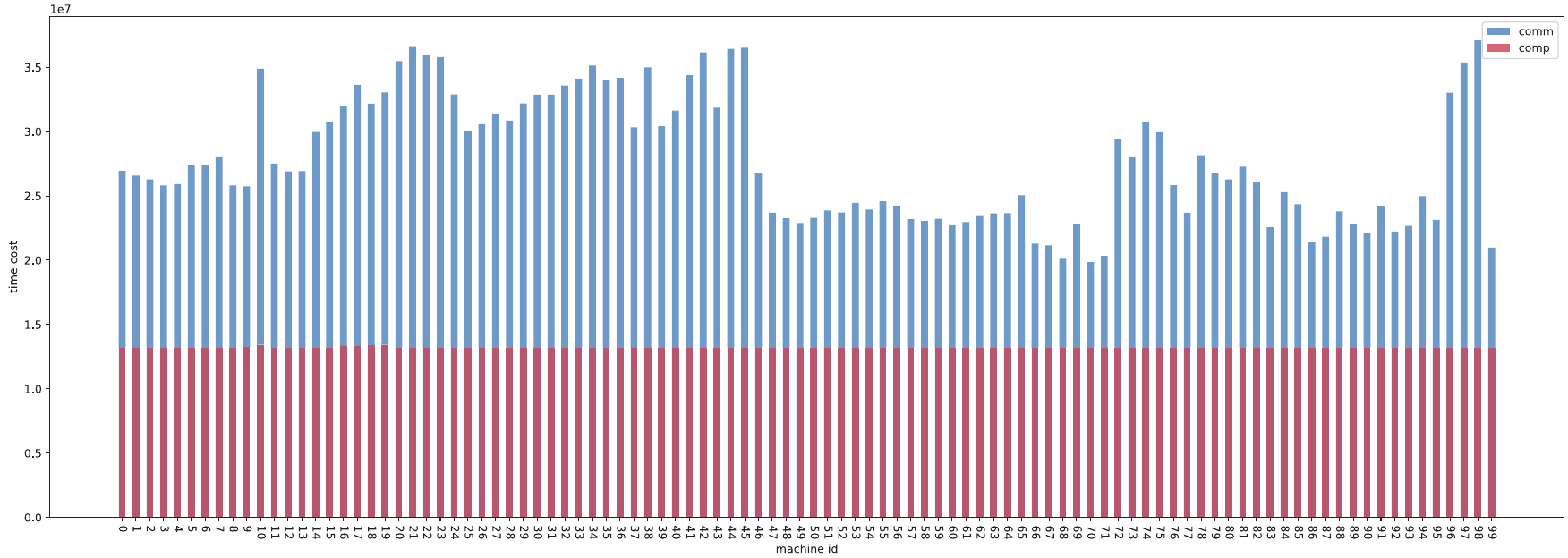}       
    \end{minipage}  
    }
    \subfigure[{\tiny WindGP}]
    {
        \label{fig:hist4-CO}
    \begin{minipage}[c]{4cm}        
        \centering      
        \includegraphics[width=4cm]{\picfolder 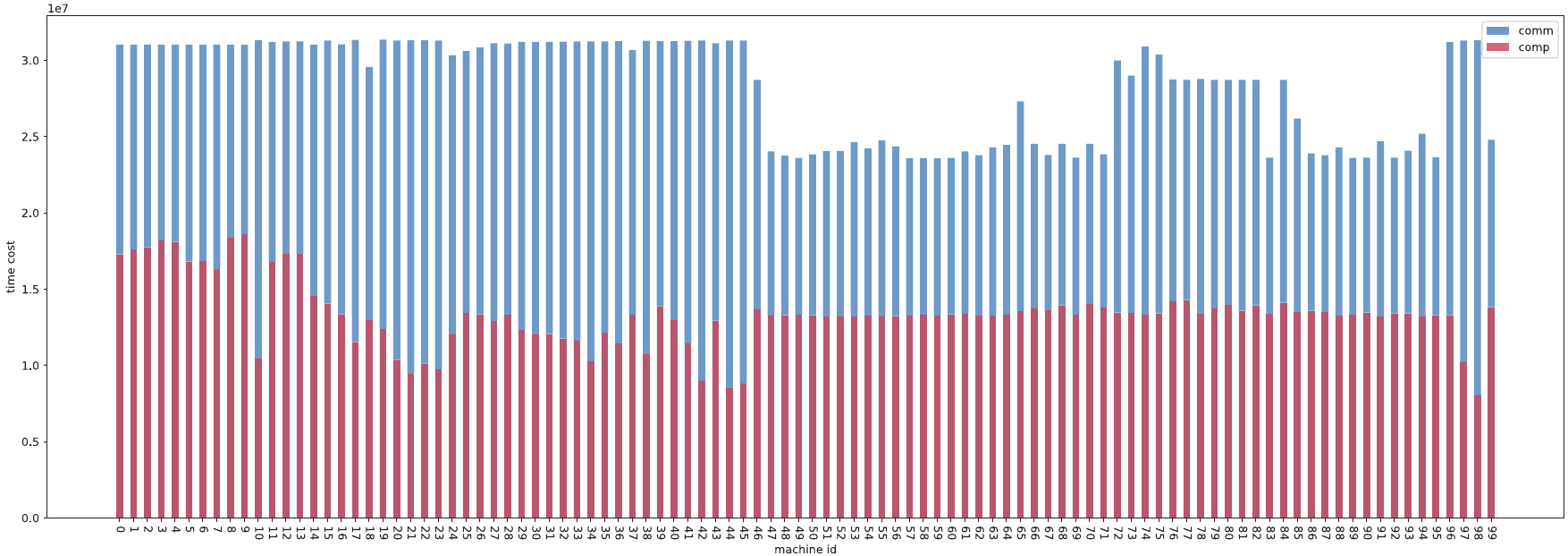}       
    \end{minipage}  
    }
    \vspace{-0.1in}
    \caption{The histogram of partitions of WindGP on CO}
	\label{fig:hist-CO}  
\end{figure*}
%}

\begin{figure}[htbp]
	\centering
	\includegraphics[width=8cm]  {\picfolder 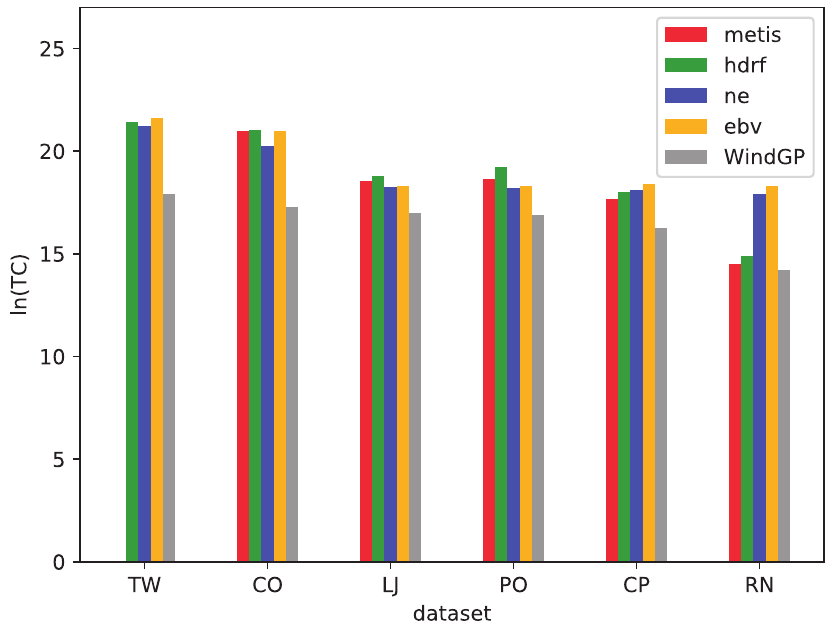}        	
    \vspace{-0.1in}     	
	\caption{Comparison of different partition algorithms (($\ln TC$ is the logarithm of $TC$))}      	
	\label{fig:tc-all}  
\end{figure}

\Paragraph{Comparison with counterparts}.
Figure \ref{fig:tc-all} lists the comparison of WindGP (the entire solution) with the state-of-the-art algorithms on heterogeneous machines.
As reported in previous work \cite{NE}, METIS can not partition too large graphs like TW due to its prohibitive memory occupation.
Among four counterparts, NE performs the best on most datasets except for CP and RN.
The two low-degree graphs have relatively larger computation cost and lower communication cost, thus the expansion scheme of NE can not optimize a lot.
In contrast, METIS and HDRF perform better on CP and RN due to their elaborated consideration of computation cost in the partition objective.
As for EBV solution, though it is carefully designed to reduce the impact of skew on power-law graphs, its objective is not equivalent to the maximum total cost and it can not be well adapted to heterogeneous cases.

WindGP outperforms all counterparts by $>3.7\times$ speedup on all power-law graphs.
Especially, on enormous social network like TW and CO, the improvement is $1\sim2$ orders of magnitudes.
This fantastic effect is produced by the perfect combination of three novel techniques in WindGP, which not only address the difference between heterogeneous machines, but also leverage the characteristics of graphs. 
The improvement is more prominent when the graph is larger and more skewed, because higher computation and communication cost implies larger optimization space.
Though the performance of WindGP is limited on mesh-like graphs like RN, it still brings $1.35\times$ speedup.
To sum up, WindGP shows $1.35\times\sim27\times$ speedup compared with the state-of-the-art partition algorithms.

Due to the lightweight preprocessing and limited iterations in post-precessing, the complexity of WindGP is linear to the graph size and machine number, which is similar to NE.
As for the executing time of partition methods, Table \ref{tab:part-time} shows that their executing time is in the same scale, while WindGP is 11\% slower than NE, which is acceptable in real cases as partitioning is not time sensitive.
Besides, extra experiments (Table \ref{tab:homo-perf}) with the best two counterparts (HDRF and NE) show that the performance of WindGP on homogeneous clusters is not worser than others, as Section \ref{sec:problem} proves that the $TC$ metric is equivalent to load balance and $RF$ in homogeneous scenarios.
Furthermore, the running time of distributed graph algorithms on the partition results is also evaluated (see Section \ref{sec:distributed}), which corresponds to the comparative results of $TC$.

\begin{table}[htbp]
	\small
	\caption{Evaluation of PageRank on homogeneous 30-machine cluster on LJ}
	\label{tab:homo-perf}
	\vspace{-0.1in}
	\begin{threeparttable}
		\small
		\centering
			\begin{tabular}{ |c|c|c|c|c|} 
				\hline
				Alg. &  $\alpha^{\prime}$ & $RF$ & $TC$ & time (s) \\
				\hline
				% 143M in heterogeneous
				HDRF  & 1.1 & 3.33 & 105M & 92 \\ 
				\hline
				% 86M in heterogeneous
				NE  & 1.1  & 1.55 & 19M & 18 \\ 
				\hline
				WindGP  & 1.0  & 1.56  & 20M & 18  \\ 
				\hline
			\end{tabular}
			%}
	\end{threeparttable}
\end{table}

% WindGP is 11\% slower than NE
\begin{table}[htbp]
	\small
	\caption{Evaluation of partitioning time (s) on traditional methods}
	\label{tab:part-time}
	\vspace{-0.1in}
	\begin{threeparttable}
		\small
		\centering
			\begin{tabular}{ |c|c|c|c|c|c| } 
				\hline
				Dataset &  METIS & HDRF & NE & EBV & WindGP \\
				\hline
				CO  & 200 & 15 & 80 & 91 & 89 \\ 
				\hline
				LJ  & 71  & 5 & 23 & 30 & 25 \\ 
				\hline
				PO  & 68  & 4 & 21 & 24 & 23 \\
				\hline
				CP  & 32  & 3 & 12 & 14 & 13 \\
				\hline
				RN  & 6  & 1 & 2 & 2.3 & 2.1 \\
				\hline
			\end{tabular}
			%}
	\end{threeparttable}
\end{table}

\subsection{Scalability Test}\label{sec:scalability}

\Paragraph{Scalability with the graph size}.
To evaluate the scalability of graph size, we use the same configuration of machines as on Twitter.
R-MAT generator \cite{R-MAT} is used to generate eight scale-free graphs, with the number of edges range from $\sim4\times10^6$ to $\sim5\times10^8$ (nearly double each time).
The parameters  of R-MAT follow the setting of Graph 500 \cite{graph500}.
The ratio of the graph’s edge count to its vertex count (i.e., half the average degree of a vertex in the graph) is 16.
The ``scale'' factor  is the logarithm base two of the number of vertices  and it starts from 18 and increases by 1 each time.
Details of these synthetic graphs are listed in Table \ref{tab:graph500}.

%\nop{
\begin{table}[htbp]
    \small
    \caption{Statistics of Graph 500 Datasets}
    \label{tab:graph500}
    \vspace{-0.1in}
    \begin{threeparttable}
        \small
        \centering
    \setlength{\tabcolsep}{3mm}{
        \begin{tabular}{crrrr}
            \toprule
            Name & $|V|$ & $|E|$ & MD\tnote{1} & Type\tnote{2} \\
            \midrule
	S18 & 262,144 & 3,800,348 & 25,707 & s \\ 
	S19 & 524,288 & 7,729,675 & 41,358 & s \\ 
	S20 & 1,048,576 & 15,680,861 & 67,086 & s \\ 
	S21 & 2,097,152 & 31,731,650 & 107,400 & s \\ 
	S22 & 4,194,304 & 64,097,004 & 170,546 & s \\ 
	S23 & 8,388,608 & 129,250,705 & 272,176 & s \\ 
	S24 & 16,777,216 & 260,261,843 & 431,690 & s \\ 
	S25 & 33,554,432 & 523,467,448 & 684,732 & s \\ 
            \bottomrule
        \end{tabular}
    }
        \begin{tablenotes}
            \item[1] Maximum degree of the graph.
            \item[2] Graph type: r:real-world, s:scale-free, and m:mesh-like.
        \end{tablenotes}
    \end{threeparttable}
\end{table}
%}

    %\vspace{-0.1in}     	
\begin{figure}[htbp]
	\centering
	\includegraphics[width=8cm]  {\picfolder 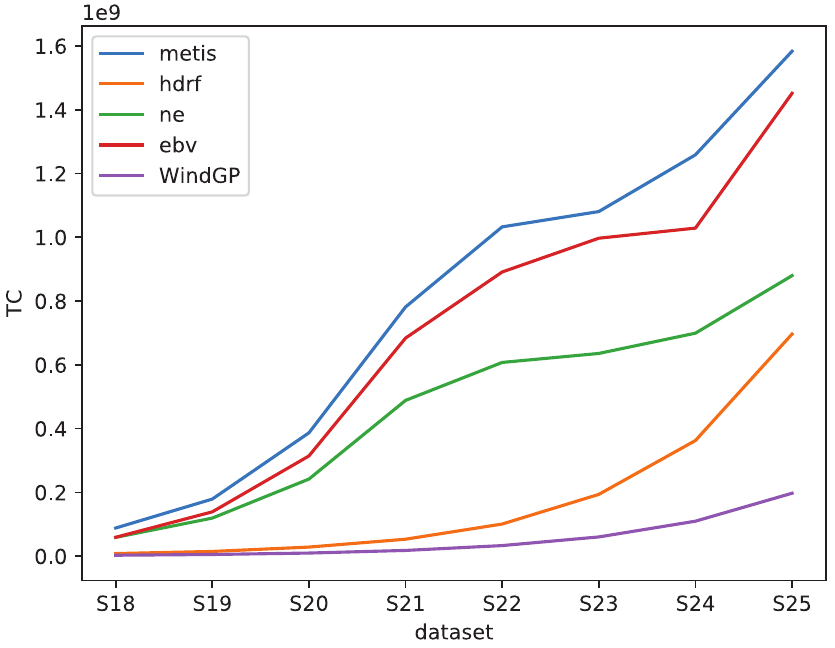}        	
    \vspace{-0.1in}     	
	\caption{The scalability with Graph 500 datasets}      	
	\label{fig:scala-g500}  
\end{figure}
   % \vspace{-0.1in}     	

The experimental result is shown in Figure \ref{fig:scala-g500}.
Obviously, the performance of WindGP is still the best when transferring to Graph 500 datasets.
WindGP not only has the minimum $TC$ on all graphs (S18 $\sim$ S25), but also displays the slowest growth as the graph size increases.
Generally, the curve of WindGP has a $\leq 1.8$ slope.
In contrast, the curves of other algorithms all show $>2$ slope.
Though the average degree remains the same, larger generated graphs are more skewed and has larger communication cost when partitioned.
Our best-first search as well as the post-processing play a significant role when reducing the communication cost.
As a result, WindGP performs well even on graphs with billions of edges.
% HDRF is good because graph 500 has low degree; metis is bad because graph 500 is high power-law

\Paragraph{Scalability with the machine number}.
We select LJ dataset and vary the machine number from 30 to 90 (increasing by 15 for each).
The ratio of super machines is all set to $\frac{1}{3}$.
Figure \ref{fig:scala-machine-num} shows the result.
Generally, $TC$ drops as the machine number increases, because more resources can be utilized to balance the cost.
However, the drop tends to be tiny when the machine number is larger than 30 (called \emph{saturation point}), which is the default setting for LJ.
Inherently, the saturation point is decided by the number of communities in the dataset.
Though more machines can help reduce the computation cost, the communication cost also rises if high-cohesion subgraph is separated into different machines.
For the same reason, NE and EBV do not fully utilize these machines, as they tend to place as many edges as possible in several super machines to minimize the communication cost.
In contrast, WindGP performs well on all settings.
If the machine number is lower than saturation point, WindGP utilizes more machines to store major subgraphs; otherwise, it treats added machines as temporary buffers for post-processing.

   % \vspace{-0.1in}     	
\begin{figure}[htbp]
	\centering
	\includegraphics[width=8cm]  {\picfolder 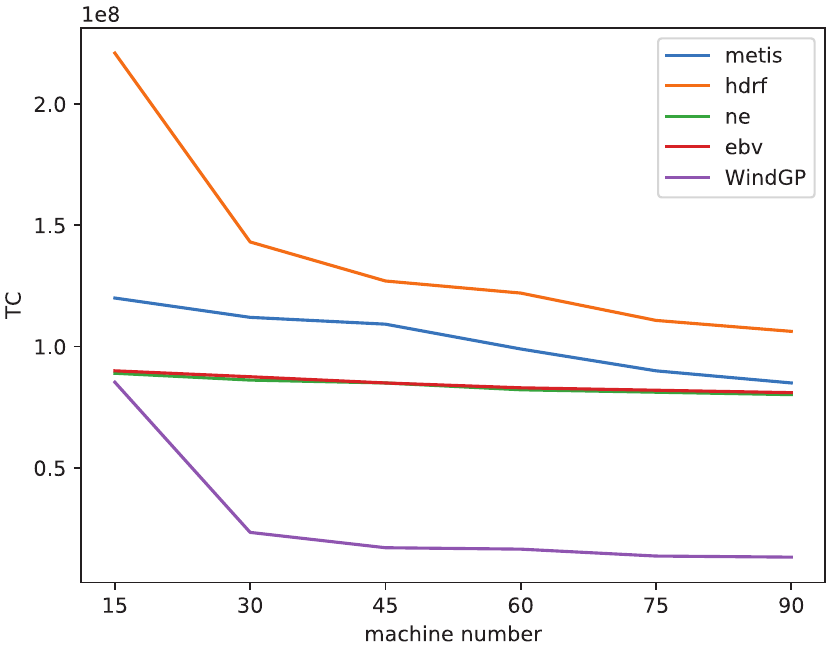}        	
    \vspace{-0.1in}     	
	\caption{The scalability with the machine number on LJ}      	
	\label{fig:scala-machine-num}  
\end{figure}
    %\vspace{-0.1in}     	

\Paragraph{Scalability with the number of machine types}.
In previous experiments, we use two kinds of machines by default.
Here we vary the number of machine types from one to six and test the performance on LJ with 30 machines.
The added type is extracted from normal machines, increase the machine memory, computation and communication cost.
Note that the 2-type setting here is different from the default setting.
Compared to the single-type setting (i.e., homogeneous cases where all machines are normal machine), 2-type setting transforms 5 machines to slightly bigger machine with larger computation and communication cost as well as larger memory size.

Generally, $TC$ increases as the number of machine types grows, i.e., the homogeneous cases achieve the minimum $TC$ for all solutions.
Note that NE performs extremely good on homogeneous cases, even better than WindGP.
However, both NE and EBV fail to adapt to heterogeneous cases and their $TC$ rises sharply when the number grows.
Comparatively, WindGP achieves the slowest growth due to its flexible preprocessing technique, which can pre-computes appropriate edge capacities for all machine settings.
Thus, WindGP is not sensitive to the number of machine types.
%BETTER: more thorough analysis, and more clear opinion (WindGP not sensitive to number of machine types)

    %\vspace{-0.1in}
\begin{figure}[htbp]
	\centering
	\includegraphics[width=8cm]  {\picfolder 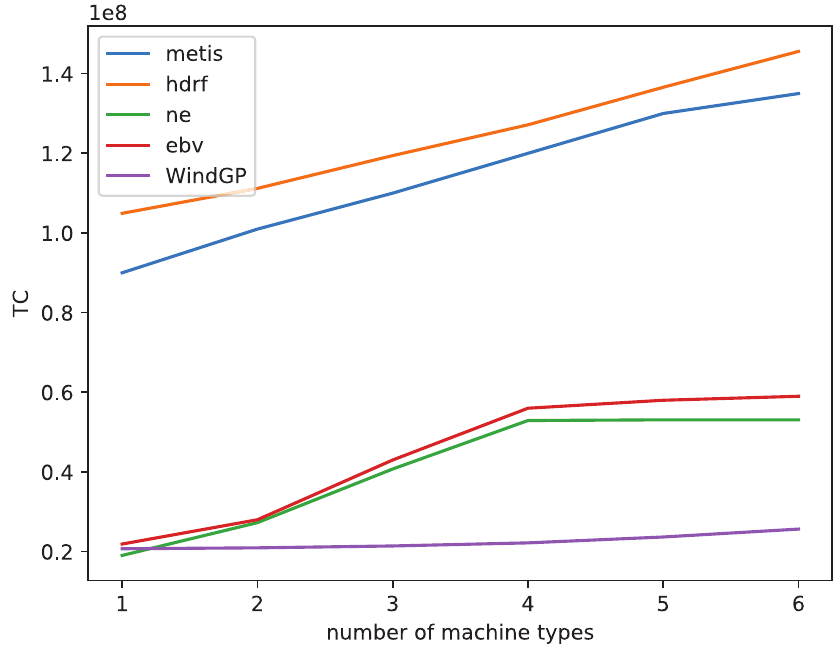}        	
    \vspace{-0.1in}     	
	\caption{The scalability with the number of machine types on LJ}      	
	\label{fig:scala-machine-num}  
\end{figure}
    %\vspace{-0.1in}

%\vspace{-0.1in}
\begin{table*}[htbp]
	\small
	\caption{Distributed running time of heterogeneous algorithms (unit:s)}
	\label{tab:hetero-perf}
	\vspace{-0.15in}
	\begin{threeparttable}
		\centering
		\begin{tabular}{|c|r|r|r|r|r|r|r|r|r|r|r|r|}
			\hline
			\multirow{2}*{DataSet} & \multicolumn{6}{|c|}{PageRank} & \multicolumn{6}{|c|}{SSSP} \\ 
			\cline{2-13}
			~ &  \cite{HeterCompPart} & GrapH & HaSGP & HAEP &  WindGP & speedup & \cite{HeterCompPart} & GrapH  & HaSGP & HAEP& WindGP & speedup   \\ 
			\hline
			TW & 1287 & 681 & 1,015 & 529 & 353 & 1.49$\times$ &  667 & 309 & 618 & 254 & 182 & 1.39$\times$   \\
			\hline
			DB & 1168 & 609 & 970 & 501 & 280 & 1.78$\times$ & 621 & 305& 597 & 226 & 158 & 1.43$\times$ \\
			\hline
			FR & 2,017 & 2,259 & 2,380  & 1,307 & 681 & 1.92$\times$ & 1,132 & 1,218 & 1,846 & 580 & 312 & 1.86$\times$ \\
			\hline
			YH & 3,204 & 3,980 & 3,012 & 2,201 & 1,048 & 2.1$\times$ & 1.910 & 2,203 & 2,018 & 1,038 & 509 & 2.04$\times$ \\
			\hline
		\end{tabular}
	\end{threeparttable}
\end{table*}
% \vspace{-0.1in}

\subsection{Evaluation of Distributed Graph Algorithms}\label{sec:distributed}
% BETTER: GNN algorithms

To evaluate the real performance of distributed graph computing, we choose PageRank \cite{PageRank} and Single-Source Shortest Path (\cite{SSSP}) and perform the test on Plato system \cite{site:plato}.
These two algorithms are the most representative dense and sparse algorithms, respectively.
In PageRank, all edges are computed and all nodes are updated in each iteration.
In contrast, the number of active nodes as well as edges grows first until the largest point, then reduces until the end.
Modern machines have large memory and computing power, thus graphs with billions of edges are added in this section.
Except for TW, referring to \cite{HaSGP} and \cite{SGSI}, we add DB (0.23B nodes and 1.1B edges), FR (65M nodes, 1.8B edges), and YH (0.41B nodes and 2.8B edges).
Similar to \cite{NE}, we use nine machines connected with Gigabit Ethernet, and $\frac{1}{3}$ of them are super machines. % (i.e., 3 super machines and 6 normal machines).
The resource of the real cluster is listed below:
\begin{itemize}
	\item 3 super machines: 1.6GHz (4), 6GB memory, 100Gbps %network bandwidth
	\item 6 normal machines: 1.6GHz (8), 2GB memory, 150Gbps %network bandwidth
\end{itemize}
The memory capacity, computation cost and communication cost of each machine ($M_i$, $C_i^{node}$, $C_i^{edge}$ and $C_i^{com}$) is estimated according to the quantification process in Section \ref{sec:problem}.
Note that the machine configuration is different from previous settings, thus $TC$ is also evaluated.
Table \ref{tab:tc-distributed}, \ref{tab:real-perf}, \ref{tab:large-real-perf} and \ref{tab:small-hetero-perf} list the comparison of $TC$ and distributed running time (seconds).

   % \vspace{-0.1in}
\begin{table}[htbp]
	\small
	\caption{The $TC$ metric on nine machines}
	\label{tab:tc-distributed}
	\vspace{-0.1in}
	\begin{threeparttable}
		%\small
		\centering
			\begin{tabular}{ |c|c|c|c| } 
				\hline
				Dataset &  HDRF & NE & WindGP \\
				\hline
				TW  & 2,790,667,265 & 5,649,273,080 & 401,478,360 \\ 
				\hline
				CO   & 938,849,480 & 565,506,575 & 189,439,055 \\ 
				\hline
				LJ   & 182,987,265 & 205,252,940 & 57,918,640 \\ 
				\hline
				PO   & 278,052,160 & 196,995,825 & 50,954,100 \\ 
				\hline
				CP   & 106,732,065 & 195,154,235 & 31,504,615 \\ 
				\hline
				RN  & 9,632,230 & 61,151,165 & 4,950,860 \\ 
				\hline
			\end{tabular}
			%}
	\end{threeparttable}
\end{table}
%\vspace{-0.1in}

 %   \vspace{-0.1in}
\begin{table}[htbp]
	\small
	\caption{Performance of distributed graph computing (unit:s)}
	\label{tab:real-perf}
	\vspace{-0.1in}
	\begin{threeparttable}
		\centering
		\begin{tabular}{|c|r|r|r|r|r|r|}
			\hline
			\multirow{2}*{Data} & \multicolumn{3}{|c|}{PageRank} & \multicolumn{3}{|c|}{TriangleCount} \\ 
			\cline{2-7}
			~ &  HDRF & NE &  WindGP & HDRF & NE & WindGP   \\ 
			\hline
			TW & 2,380 & 4,417 & 353 &  1,046 & 1,520 & 182   \\
			\hline
			CO & 791 & 398 & 125 & 325 & 147 & 97  \\
			\hline
			LJ & 112 & 201 & 54  & 33 & 33 & 25  \\
			\hline
			PO & 130 & 187 & 41 & 29 & 22 & 18  \\
			\hline
			CP & 64 & 166 & 27 & 11 & 12 & 11  \\
			\hline
			RN & 11 & 52 & 9 & 7 & 9 & 5  \\
			\hline
		\end{tabular}
	\end{threeparttable}
\end{table}
%\vspace{-0.1in}

   %\vspace{-0.1in}
\begin{table*}[t]
	\small
	\caption{Performance of distributed graph computing (unit:s)}
	\label{tab:large-real-perf}
	\vspace{-0.15in}
	\begin{threeparttable}
		\centering
		\begin{tabular}{|c|r|r|r|r|r|r|r|r|r|}
			\hline
			\multirow{2}*{DataSet} & \multicolumn{3}{|c|}{$TC$}  & \multicolumn{3}{|c|}{PageRank} & \multicolumn{3}{|c|}{SSSP} \\ 
			\cline{2-10}
			~ & HDRF & NE & WindGP &  HDRF & NE &  WindGP & HDRF & NE & WindGP   \\ 
			\hline
			TW & 2.7G & 5.6G & 0.4G & 2,380 & 4,417 & 353 &  1,046 & 1,820 & 182   \\
			\hline
			DB & 2.5G & 4.2G & 0.3G & 2,146 & 3,709 & 280 & 965 & 1,227 & 158  \\
			\hline
			FR & 5.1G & 4.4G & 1.0G & 5,012 & 4,183 & 681  & 1,398 & 1,260 & 312  \\
			\hline
			YH & 6.8G & 5.9G & 1.7G & 8,149 & 5,980 & 1,048 & 2,012 & 1,865 & 509  \\
			\hline
		\end{tabular}
	\end{threeparttable}
\end{table*}
% \vspace{-0.1in}

\Paragraph{Comparison with Non-heterogeneous Solutions}.
Compared with non-heterogeneous solutions, WindGP has the lowest $TC$ value on all datasets, and the improvement is rather prominent (>6$\times$) on enormous skewed graphs (TW and DB, the maximum degree is 3M and 17M respectively).
FR and YH have relatively much smaller skew (the maximum degree is 5.2K and 2.5K respectively), thus the speedup of WindGP is restricted to $3\times \sim 4\times$.
For most cases, a reduction of $TC$ (i.e., the estimated total cost) means a reduction in the running time of distributed graph computing as the partition quality is improved. 
This further verifies the analysis in Section \ref{sec:problem} that $TC$ is nearly proportional to the distributed running time.
Though the optimization effect of our techniques is tiny on mesh-like graphs like RN, the running time still drops a little thanks to our sophisticated post-processing technique which considers all kinds of corner cases.
The running time of both PageRank and SSSP is much larger than that on homogeneous machines, because the distributed computing on heterogeneous machines is much more complicated.
Super machines have large memory but high computation and communication cost, while the case of normal machines is reversed.
This limits the optimization space if all normal machines are full, which corresponds to the analysis in Section \ref{sec:scalability}.

\Paragraph{Comparison with Heterogeneous Solutions}.
As shown in Table \ref{tab:hetero-perf}, when compared with heterogeneous solutions (\cite{HeterCompPart}, GrapH, HaSGP, HAEP), WindGP also shows $>1.49\times$ and $>1.39\times$ speedup in the running time of PageRank and SSSP respectively, as it performs collaborative optimization on machines with heterogeneous memory capacity, computing power and network bandwidth.
Note that the speedup is calculated by comparing WindGP with the best counterpart HAEP.
None of the counterparts takes care of the heterogeneous memory capacity, thus they can not compete with WindGP when optimizing the computation cost.
Specifically, \cite{HeterCompPart} only optimizes load balance while \cite{GrapH} targets at communication cost.
Thus, the computing time of \cite{HeterCompPart} is similar to ours, but its communication time is $\sim$50\% longer.
\cite{GrapH} has >20\% longer computing time, and its communication time is also $\sim$18\% longer as it lacks the best-first search and comprehensive subgraph-local tuning.
HaSGP is a streaming partition algorithm, which performs terrible when running on severely skewed graphs like TW and DB.
The degree distribution of FR and YH is much more balanced, thus \cite{HeterCompPart} performs better than GrapH as it balances the computation cost better while existing communication-optimized methods (NE, GrapH, HAEP) have non-prominent effect on these two graphs.
Inherently, most nodes have the same scale of neighbors, and there is no significant difference when we select which one as the boundary.
However, WindGP can optimize this case well because the border generation in Section \ref{sec:best-search} restricts the border vertex to reside in as fewer machines as possible.

%\vspace{-0.1in}
\begin{table}[htbp]
	\small
	\caption{Distributed running time of heterogeneous algorithms (unit:s)}
	\label{tab:small-hetero-perf}
	\vspace{-0.1in}
	\begin{threeparttable}
		\centering
		\begin{tabular}{|c|r|r|r|r|r|r|}
			\hline
			\multirow{2}*{Data} & \multicolumn{3}{|c|}{PageRank} & \multicolumn{3}{|c|}{TriangleCount} \\ 
			\cline{2-7}
			~ &  \cite{HeterCompPart} & GrapH &  WindGP & \cite{HeterCompPart} & GrapH & WindGP   \\ 
			\hline
			TW & 1287 & 681 & 353 &  667 & 309 & 182   \\
			\hline
			CO & 411 & 208 & 125 & 198 & 116 & 97  \\
			\hline
			LJ & 98 & 81 & 54  & 32 & 31 & 25  \\
			\hline
			PO & 201 & 126 & 41 & 38 & 27 & 18  \\
			\hline
			CP & 53 & 108 & 27 & 18 & 20 & 11  \\
			\hline
			RN & 10 & 32 & 9 & 7 & 21 & 5  \\
			\hline
		\end{tabular}
	\end{threeparttable}
\end{table}
%\vspace{-0.1in}

Comparing PageRank and SSSP, obviously the speedup on PageRank is always higher than that on SSSP.
The biggest difference is on DB, which has the largest degree skew and the lowest averaged degree (i.e., 4.7) among four large graphs.
In PageRank, all nodes need to join each iteration while only partial nodes are active in SSSP.
The largest skew enlarges the optimization space of WindGP in PageRank, while the lowest averaged degree limits our techniques in SSSP (just like the cases in mesh-like RN).
This indicates that WindGP works well on dense algorithms like PageRank, but further improvement can be developed on sparse algorithms like SSSP.
% better on pagerank than tc/sssp? more heavy alg has larger compute and more comm?
% dataset: mesh has little comm, especially for tc
%More experimental details are omitted due to space limit, which can be supplemented in the full version if required.
%$TC$ is not compared for these solutions

To sum up, compared with existing non-heterogeneous and heterogeneous solutions, WindGP successfully reduces the $TC$ metric as well as the running time of distributed graph computing for all graphs.
On average, WindGP reduces $TC$ and the distributed running time by $1.9\times\sim6.7\times$ and $1.4\times\sim5.7\times$, respectively.
As for the executing time of heterogeneous partition methods, Table \ref{tab:part-time} shows that the gap is not prominent.

\begin{table}[htbp]
	\small
	\caption{Evaluation of partitioning time on heterogeneous methods}
	\label{tab:part-time-large}
	\vspace{-0.1in}
	\begin{threeparttable}
		\small
		\centering
			\begin{tabular}{ |c|c|c|c|c|c| } 
				\hline
				Dataset &  \cite{HeterCompPart} & GrapH & HaSGP & HAEP & WindGP \\
				\hline
				TW  & 105 & 92 & 91 & 112 & 101 \\ 
				\hline
				DB  & 109  & 90 & 90 & 110 & 108 \\ 
				\hline
				FR  & 156  & 137 & 138 & 170 & 152 \\
				\hline
				YH  & 301  & 218 & 210 & 312 & 293 \\
				\hline
			\end{tabular}
			%}
	\end{threeparttable}
\end{table}

\section{Conclusions}\label{sec:conclusion}

We introduce a graph partitioning algorithm WindGP, which supports fast and high-quality partitioning on heterogeneous machines.
WindGP consists of three phases: preprocessing, graph exploration, post-processing.
First, novel preprocessing is utilized to simplify the metric and balance the computation cost according to the characteristics of graphs and machines. 
Besides, best-first search is proposed instead of BFS/DFS, in order to generate communities with high cohesion. 
Furthermore, subgraph-local search is adopted to tune the partition results adaptively.
Experiments on real-world graphs show that WindGP outperforms the state-of-the-art partition methods by 1.35$\times$$\sim$27$\times$.
In future, targeting at specific graph algorithms (e.g., PageRank), WindGP can be further enhanced by utilizing algorithmic characteristics.

%\begin{acks}
%	This work was supported by the [...] Research Fund of [...] (Number [...]). Additional funding was provided by [...] and [...]. We also thank [...] for contributing [...].
%\end{acks}

\clearpage

\bibliographystyle{ACM-Reference-Format}
\bibliography{windgp}

\end{document}